\pdfoutput=1
\documentclass[]{article}  
\usepackage{url,float}
\usepackage{graphicx}
\usepackage{amsmath}
\usepackage{amsfonts}
\usepackage{amssymb}
\usepackage{latexsym}
\usepackage{xcolor}
\usepackage[colorlinks]{hyperref}
\hypersetup{
    colorlinks=true,
    linkcolor=blue,
    filecolor=magenta,      
    urlcolor=cyan,
    citecolor=purple
    }

\newcommand{\hide}[1]{}

\newcommand{\ABox}{
\raisebox{3pt}{\framebox[6pt]{\rule{6pt}{0pt}}}
}
\newenvironment{proof}{{\bf Proof:}}{\hfill\ABox}

\newtheorem{thm}{{\bf Theorem}}
\newtheorem{cor}{Corollary}
\newtheorem{lem}{Lemma}
\newtheorem{prop}{Proposition}
\newtheorem{rmk}{Remark}

\newtheorem{ex}{Example}
\newtheorem{op}{Open Problem}

\newcommand{\lemlab}[1]{\label{lemma:#1}}
\newcommand{\thmlab}[1]{\label{therom:#1}}
\newcommand{\exlab}[1]{\label{ex:#1}}
\newcommand{\proplab}[1]{\label{prop:#1}}

\newcommand{\figlab}[1]{\label{fig:#1}}
\newcommand{\seclab}[1]{\label{sec:#1}}
\newcommand{\rmklab}[1]{\label{rmk:#1}}

\newcommand{\lemref}[1]{\ref{lemma:#1}}
\newcommand{\thmref}[1]{\ref{therom:#1}}

\newcommand{\propref}[1]{\ref{prop:#1}}
\newcommand{\exref}[1]{\ref{ex:#1}}
\newcommand{\secref}[1]{\ref{sec:#1}}
\newcommand{\figref}[1]{\ref{fig:#1}}

\newcommand{\rmkref}[1]{\ref{rmk:#1}}

\def\a{{\alpha}}
\def\b{{\beta}}
\def\p{{\phi}}
\def\q{{\theta}}

\def\C{{\mathcal C}}

\def\g{{\gamma}}

\def\d{{\delta}}
\def\D{{\Delta}}
\def\e{{\varepsilon}}
\def\I{{\iota}}
\def\conv{\mathop{\rm conv}\nolimits}
\def\R{{\mathbb{R}}}
\def\T{{\mathcal T}}
\def\Sk{{\operatorname {Sk}}}

\newcommand{\squeezelist}{\setlength{\itemsep}{0pt}}

\usepackage[shortlabels]{enumitem}

\title{Skeletal Cut Loci on Convex Polyhedra\footnote{
A preliminary version of this paper 
(excluding trees with degree-$2$ nodes and partially skeletal cut loci)
was presented at a conference~\cite{SkCL-CCCG}.}
}

\author{
Joseph O'Rourke
\and
Costin V\^\i lcu
}

\date{\today}

\begin{document}
\maketitle

\begin{abstract}
On a convex polyhedron $P$, the \emph{cut locus} $\C(x)$ with respect to a point $x$
is a tree of geodesic segments (shortest paths) on $P$ that includes every vertex. 
We say that $P$ has a \emph{skeletal cut locus} if there is some $x \in P$
such that  $\C(x) \subset \Sk(P)$, where  $\Sk(P)$ is the \emph{$1$-skeleton} of $P$.
At a first glance, there seems to be very little relation between the cut locus and the $1$-skeleton,
as the first one is an intrinsic geometry notion, and the second one specifies the combinatorics of $P$.

In this paper we study skeletal cut loci, obtaining four main results.
First, given any combinatorial tree $\T$,
there exists a convex polyhedron $P$ and a point $x$ in $P$
with a cut locus that lies in $\Sk(P)$, and whose combinatorics match $\T$.
Second, any (non-degenerate) polyhedron $P$ has at most a finite number of 
points $x$ for which $\C(x) \subset \Sk(P)$.
Third, we show that almost all polyhedra have no skeletal cut locus.
Fourth, we provide a combinatorial restriction to the existence of skeletal cut loci.

Because the source unfolding of $P$ with respect to $x$ 
is always a non-overlapping net for $P$,
and because the boundary of the source unfolding is the (unfolded) cut locus,
source unfoldings of polyhedra with skeletal cut loci are edge-unfoldings, and moreover ``blooming,'' avoiding self-intersection during an unfolding process.

We also explore partially skeletal cut loci, leading to \emph{partial edge-unfoldings}; i.e., unfoldings obtained by cutting along some polyhedron
edges and cutting some non-edges.
\end{abstract}


\section{Introduction}
\seclab{Introduction}


\subsection{Background and Results}
Our focus is the cut locus $\C(x)$ of a point $x$ on a convex polyhedron $P$, 
and the relationship of $\C(x)$ to the \emph{$1$-skeleton} of $P$---
the graph of vertices and edges---which we denote by $\Sk(P)$.

The \emph{cut locus} $\C(x)$ of $x \in P$ is the closure of the set of points on $P$ to which there is more than one geodesic segment (shortest path) from $x$.
$\C(x)$ is a tree whose leaves are vertices of $P$. 
Nodes of degree $k \ge 3$ are \emph{ramification points}
to which there are $k$ distinct geodesic segments from $x$.
Nodes $v$ of degree $2$ in $\C(x)$ can also occur, if $v$ is a vertex of $P$.
For details, see Section~\secref{Preliminaries}.

The $1$-skeleton of a non-degenerate polyhedron is a $3$-connected 
planar graph by Steinitz's theorem.
We call a doubly-covered convex polygon a \emph{degenerate} convex polyhedron,
for which the $1$-skeleton is a cycle.

We say that $P$ has (or posesses)
a \emph{skeletal cut locus} if there is some $x \in P$
such that  $\C(x) \subset \Sk(P)$.
Such a polyhedron $P$ is called \emph{cut locus amenable}, \emph{amenable} for short.

The edges of $\C(x)$ are known to be geodesic segments~\cite{aaos-supa-97},
so it is at least conceivable that an edge of $\C(x)$ lies along an edge of $P$.
Theorem~\thmref{EveryTree} shows that, for certain polyhedra $P$ and points $x \in P$,
all of $\C(x)$ lies in the $1$-skeleton of $P$:  $\C(x) \subset \Sk(P)$.
As a simple example, we will see in Lemma~\lemref{EveryTetra} that the three edges incident to any
vertex $v_i$ of a tetrahedron form $\C(x)$ for an appropriate $x$,
and are therefore a skeletal cut locus.

Although Theorems~\thmref{Finite} and~\thmref{Rare} will show that skeletal cut loci are ``rare'' in senses we'll make precise, Theorem~\thmref{EveryTree}
and its proof establish that uncountably many polyhedra do admit skeletal cut loci, 
in a sense made quantitatively precise by Proposition~\propref{Uncountable}.

Theorem~\thmref{Every-vertex} characterizes those polyhedra every vertex of which has a skeletal cut locus.
Complementing its first part, Theorem~\thmref{Restriction} provides a simple combinatorial restriction to the existence of skeletal cut loci, connecting to a current topic in graph theory.

Theorem~\thmref{EveryTree} can also be viewed as a companion to the
main result in~\cite{ov-clrcp-2023}, that any \emph{length tree}---a tree
with specified edge lengths---can be realized as the cut locus on some polyhedron.
Here we only match the combinatorics of $\T$, not its metrical properties, 
but requiring additionally for $\T$ to be included in $\Sk(P)$.

\paragraph{Connection to Unfolding.}
It has long been known that cutting the cut locus $\C(x)$
and unfolding to the plane leads to the non-overlapping
\emph{source unfolding}:
If $x$ is not itself at a vertex, then the unfolding arrays all
the shortest paths $2\pi$ around $x$, with the image of the
cut locus forming the boundary of the unfolding~\cite{m-fspcp-85}~\cite{ss-spps-86}.
For the polyhedra in Theorem~\thmref{EveryTree},
the source unfolding is an edge-unfolding.
This adds another infinite class of polyhedra 
(which we call \emph{tapered} in Section~\secref{Case_d})
that are known to have edge-unfolding \emph{nets}.
And because it is known that the source unfolding can be \emph{bloomed}---unfolded continuously
from $\R^3$ to $\R^2$ without self-intersection~\cite{demaine2011continuous}---Theorem~\thmref{EveryTree} 
and its companion Proposition~\propref{Uncountable} provide  
 perhaps the first infinite class of examples of blooming edge-unfoldings.

A central open problem asks for an accounting
of all the polyhedra $P$ that support a skeletal cut locus.
All of these enjoy the property that source unfoldings  are also blooming edge-unfoldings.

In Section~\secref{Partial-Edge Unfoldings} we introduce and briefly explore
partially skeletal cut loci, leading to \emph{partial edge-unfoldings}; i.e., unfoldings obtained by cutting along some polyhedron edges and cutting some non-edges.


\subsection{Cut Locus Preliminaries}
\seclab{Preliminaries}
For the readers convenience, we list next several basic properties of cut loci,
sometimes used implicitly in the following.

\begin{enumerate}[label={(\roman*)}]
\item 
\label{i}
$\C(x)$ is a tree drawn on the surface of $P$.
Its leaves are vertices of $P$, and all vertices of $P$, excepting $x$ (if it is a vertex) are included in $\C(x)$. 
All points interior to $\C(x)$ of degree $3$ or more are known as \emph{ramification points} of $\C(x)$.
All vertices of $P$ interior to $\C(x)$ are also considered as ramification points, of degree at least $2$;
see e.g. Fig.~\figref{DiPyramid}.
\item 
\label{ii}
Each point $y$ in $\C(x)$ is joined to $x$ by as many geodesic segments\footnote{%
We will sometimes abbreviate ``geodesic segment'' by \emph{geoseg},
and ``geodesic arc'' by \emph{geoarc}.}  
as the number of connected components of $\C(x) \setminus {y}$.
For ramification points in $\C(x)$, this is precisely their degree in the tree.
\item 
\label{iii}
The edges of $\C(x)$ are geodesic segments on $P$.
\item 
\label{iv}
Assume the distinct geodesic segments $\g$ and $\g'$ 
 from $x$ to $y \in C(x)$ bound a domain $D$ of $P$, 
which intersects no other geodesic segment from $x$ to $y$.
Then there is an arc of $\C(x)$ at $y$ which intersects $D$ 
and bisects the angle of $D$ at $y$.
\item 
The tree $\C(x)$ is reduced to a path if and only if
the polyhedron is a doubly-covered (planar) convex polygon, with $x$ on the rim.
\end{enumerate}

\noindent
Further details and references can be found in~\cite[Ch.~2]{Reshaping}.


\section{Main Result and Examples}
\seclab{Overview}
Our main result is the following theorem.

\begin{thm}
\thmlab{EveryTree}
Given any combinatorial tree $\T$ 
there is a convex polyhedron $P$ and a point $x \in P$
such that the cut locus $\C(x)$ is entirely contained in $\Sk(P)$,
and the combinatorics of $\C(x)$ match $\T$.
\end{thm}

The proof of Theorem~\thmref{EveryTree} is quite long.
It consists of a case analysis (Section~\secref{deg2}), a detailed construction for each (sub)case
(Sections~\secref{ConstructionDetails}--\secref{Examples}, 
\secref{Case_a},
\secref{Case_b},
\secref{Case_c},
\secref{Case_d}), 
and a concluding induction (Section~\secref{Induction Proof}).
Rather than plunging right into the proof, we instead
continue by sketching the main proof idea, and then
presenting the consequences and implications of Theorem~\thmref{EveryTree},
postponing the proof details to 
Sections~\secref{ConstructionDetails}--\secref{Induction Proof}.

We next illustrate the main idea of the construction,
with the simple case of a tree without degree-$2$ nodes.
Suppose the given tree $\T$ is the $7$-leaf tree shown in Fig.~\figref{Tree_7leaves}.
We select a degree-$3$ node as root $a$, which corresponds to the apex
of a regular tetrahedron $a v_1 v_2 v_3$.
We fix $x$ at the centroid of the base $Q$. 
\begin{figure}[htbp]
\centering
\includegraphics[width=0.5\columnwidth]{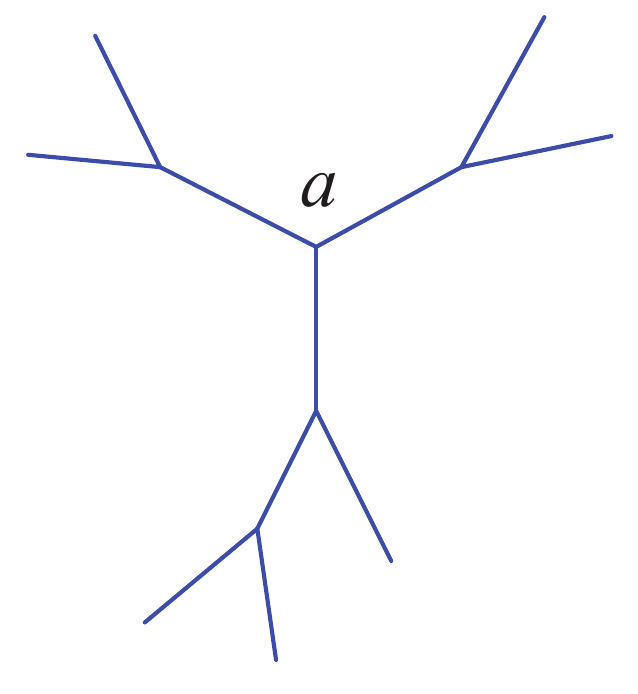}
\caption{Tree $\T$ with $7$ leaves.}
\figlab{Tree_7leaves}
\end{figure}

Fig.~\figref{T123_3d2d}(a) show one possible construction of $P$.
The edges incident to $a$ are clearly in $\C(x)$ with $x$ at the centroid of the base triangle.
All three base vertices of the tetrahedron are then truncated, with the
truncation of $v_1$ truncated a second time.
Now $T$ corresponds to all the non-base edges of $P$.

The truncations are not arbitrary:
the truncation planes must have precise tilts 
in order for the edges of each truncation to lie in $\C(x)$.
Fig.~\figref{T123_3d2d}(b) shows the source unfolding of $P$, 
with $a_1,a_2,a_3$ the three images of $a$.
The red bisector rays from $x$ through the truncation vertices on the base $Q$
suggest that indeed any point $p$ on a truncation edge is equidistant from $x$
and therefore on $\C(x)$.

\begin{figure}[htbp]
\centering
\includegraphics[width=1.0\columnwidth]{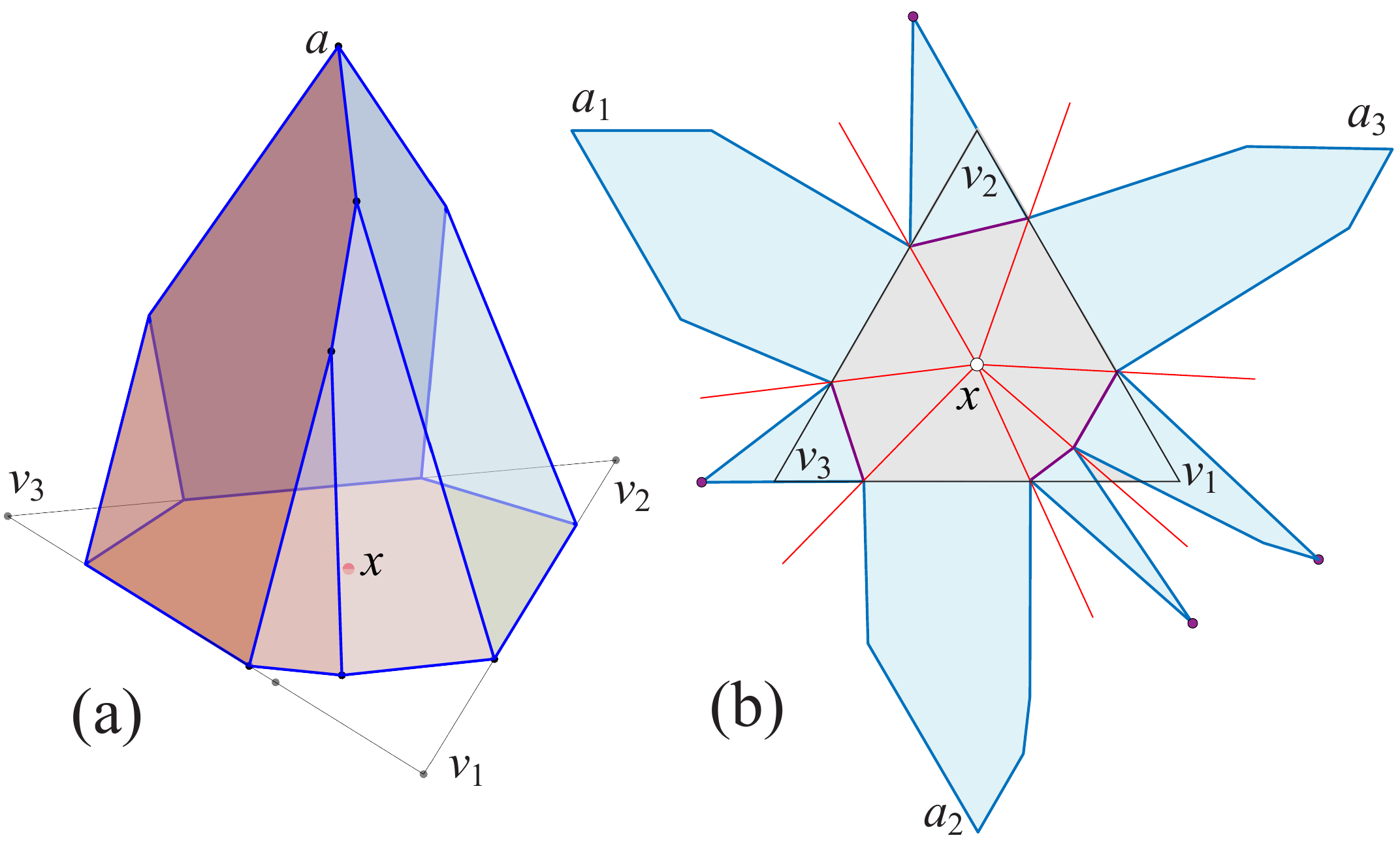}
\caption{(a)~$P$ is created from a regular tetrahedron by four vertex truncations.
$\C(x)$ consists of all non-base edges.
(b)~Source unfolding of $P$ from $x$. Bisectors shown red.}
\figlab{T123_3d2d}
\end{figure}

Returning to the need for precise tilts of the tuncation planes,
let $z$ be the point on the edge $a v_1$ through which the truncation plane passes,
creating a truncation triangle $z t_1 t_2$.
As indicated in Fig.~\figref{Tfour_3d},
the tilt is uniquely determined by the location of $z$:
the placement of $z$ determines $t_1,t_2$, and the edge $t_1 t_2$ determines $z$.

\begin{figure}[htbp]
\centering
\includegraphics[width=0.80\columnwidth]{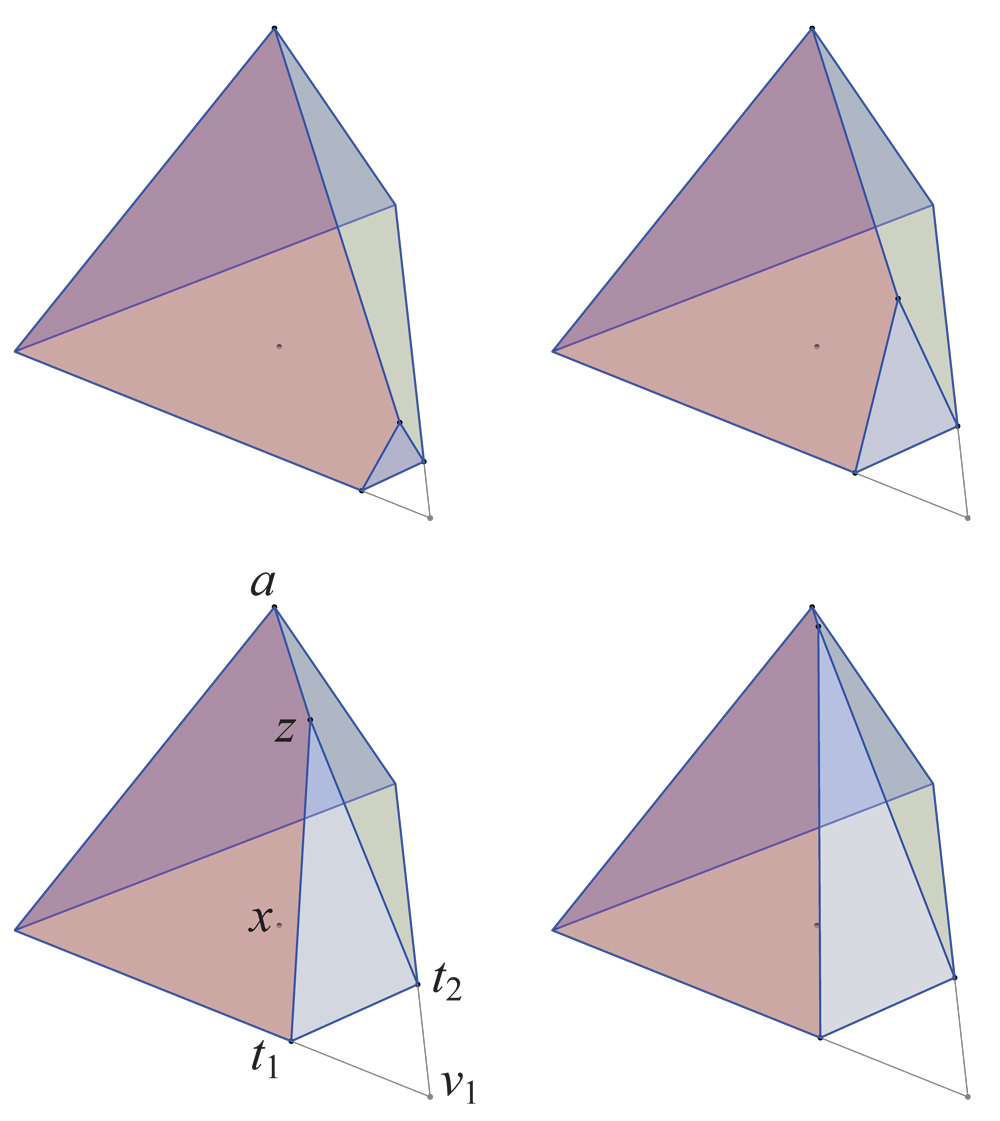}
\caption{The tilt of the truncation plane is determined by the position of $z$
on $a v_1$.}
\figlab{Tfour_3d}
\end{figure}

\medskip

Vertex truncations naturally increase the degree of a polyhedron vertex,
matching the degree of a node of $\T$. See, e.g., Fig.~\figref{k=8_nolab} in 
Section~\secref{Examples}.
One reason the proof of Theorem~\thmref{EveryTree} is so long is that
a degree-$2$ node of $\T$ cannot be created by the same basic truncation process.
Instead we found it necessary to partition the possible occurrences of degree-$2$ nodes
into four subcases (Fig.~\figref{FourCases}).


\section{Theorem~\thmref{EveryTree} Discussion}
\seclab{Theorem1Conclusions}
We mentioned in Section~\secref{Introduction} that Theorem~\thmref{EveryTree} leads
to an uncountable number of skeletal polyhedra.
This follows immediately from the freedom to place $z$ at any point interior to $a v_1$
in the construction detailed in Section~\secref{ConstructionDetails}.
We can be more quantitatively precise, as follows.

Assume that $\T$ is a cubic tree without degree-$2$ nodes, 
so it has $n$ leaves and $n-2$ ramification points.
Aside from 
one ramification point, which is chosen as the apex of the starting 
tetrahedron, all others are free to vary on their respective edges 
in our construction, which implies $n-3$ free parameters.
Because $\C(x)$ is skeletal, each ramification point of $\T$ is a vertex of $P$,
so $P$ has $V=2n-2$ vertices, and $n=V/2 +1$. 
The space ${\mathfrak P}_V$ of all convex polyhedra with $V$ vertices, up to isometries, has dimension $3V-6$ (see for example~\cite{lebedeva2022alexandrov}),
hence the starting tetrahedron provides another $6$ free parameters and
we have the next result.

\begin{prop}
\proplab{Uncountable}
The set of convex polyhedra admitting skeletal cut loci---and hence blooming edge-unfoldings---contains a subset of dimension $\geq V/2+ 4$ in the $(3V{-}6)$-dimensional space of all convex polyhedra with $V$ vertices, up to isometries. 
\end{prop}

Our construction for trees without degree-$2$ nodes in Theorem~\thmref{EveryTree}
(see Sections~\secref{Overview} and \secref{ConstructionDetails})
results in a \emph{dome}, a convex polyhedron $P$ with a distinguished base face $Q$,
with every other face sharing an edge with $Q$.
It was already known that domes have edge-unfoldings~\cite[p.~325]{do-gfalop-07}, 
although the proof of non-overlapping for our domes is almost trivial---the source
unfolding does not overlap.

However, there are many other polyhedra with skeletal cut loci,
see, e.g., Fig.~\figref{DiPyramid}, Theorem~\thmref{Every-vertex},
Fig.~\figref{StackedPyramids}, and Section~\secref{Tapered}.
Which leaves us with this central open problem:
\emph{Characterize all convex polyhedra $P$ which admit skeletal cut loci},
i.e., characterize the \emph{amenable} convex polyhedra.
The remainder of the paper, before giving the main proof, 
addresses and partially answers this problem.

\medskip
Several natural questions now suggest themselves:
\begin{enumerate}[(1)]
\item For a fixed $P$, how many distinct points $x$ can lead to skeletal cut loci?
(Theorem~\thmref{Finite}).
\item Can all of $\Sk(P)$ for a given $P$ be covered by several cut loci? (Proposition~\propref{Degenerate}).
\item How common / rare are skeletal cut loci in the space of all convex polyhedra?
(Theorem~\thmref{Rare}).
\item Are there restrictions for the existence of skeletal cut loci?
(Proposition~\propref{Degenerate}, Theorems~\thmref{Finite} and~\thmref{Restriction}).
\end{enumerate}


\section{Existence of Several Skeletal Cut Loci}
\seclab{Several}

\noindent
In the first two questions in the list above, degenerate $P$ play a special role:

\begin{prop}
\proplab{Degenerate} 
\begin{enumerate}[(a)] 
\item There exists infinitely many points $x$ with $\C(x) \subset \Sk(P)$ 
if and only if $P$ is degenerate.
\item There exists two points $x_1,x_2$ on $P$ whose cut loci together cover $\Sk(P)$ 
if and only if $P$ is degenerate.
\end{enumerate}
\end{prop}

The finitness claim in Theorem~\thmref{Finite} is then a corollary of claim (a).
Before arguing for a quantitative statement of this theorem, we make two
observations. First, for $P$ degenerate, $\Sk(P)$ is the rim of $P$,
and for any $x$ on the rim, $\C(x)$ is a subset of $\Sk(P)$.
So one direction (a) of the proposition is trivial.
Second, a special case asks whether it could be that each
vertex $v$ of $P$ leads to a skeletal cut locus $\C(v)$.
The answer is \textsc{yes}, realized, for example, by the regular octahedron.

\begin{thm}
\thmlab{Finite}
For any non-degenerate $P$ with $E$ edges,
there are at most $2 {E \choose 2}$ flat points $x$ of $P$ such that $\C(x) \subset \Sk(P)$.
\end{thm}
\begin{proof}  
Assume there exists a flat point $x$ of $P$, such that $\C(x) \subset \Sk(P)$.
Then $x$ belongs to one or two faces, $x \in F_j$, with $j \in \{1,2\}$.
Let $F$ denote either $F_1$ if $j=1$, or the union $F_1\cup F_2$ if $j=2$.

Denote by $v_i$, $i \geq 3$, the vertices of $F$, and by $e_i$ 
the edge of $\C(x)\subset \Sk(P)$ incident to $v_i$ and not included in $F$.
Finally, denote by $\g_i$ the geodesic segment from $x$ to $v_i$.

Because $e_i \subset C(x)$, $\g_i$ and $e_i$ together bisect the complete angle at $v_i$, 
by the bisection property~\ref{iv} of the cut locus.
In other words, the straight extensions $E_i$ into $F$ 
by all $e_i$ are concurrent: they intersect at the point $x$.

Now we count all the possible locations $x$ over all edges of $P$.
Consider a pair of edges $e_i,e_j$.
Each has possible edge extensions from each endpoint.
So the edge extensions are geodesic rays.
Two such straight extensions could intersect several times on $P$.
However, only their first intersection beyond the endpoints is a possible
location for $x$.
Each edge has two extensions, one from each endpoint, and
because there are $E$  straight extensions of the $E$ edges of $P$, there are at most $2 {E \choose 2}$
possible locations for $x$.
\end{proof}

\medskip
\noindent
So this theorem settles the other direction of Proposition~\propref{Degenerate}(a).

Now we prove Proposition~\propref{Degenerate}(b),
that only degenerate $P$ allows covering $\Sk(p)$ by 
only two cut loci.

\medskip
\noindent
\begin{proof}
If $P$ is degenerate then any two points on its rim, but not on the same edge, satisfy the conclusion.

Assume now that $P$ is non-degenerate and  $x \in P$ such that $\C(x) \subset \Sk(P)$.
Then $\C(x)$ has at least one ramification point of degree $d \geq 3$,
as it is known that only degenerate $P$ support path cut loci. 
The $d$ edges of $\C(x)$ lie in at least $3$ faces of $P$.
Then there exists a cycle in $\Sk(P)$, formed by edges of those faces which are not in $\C(x)$.
But such a cycle cannot be covered by only one other cut locus, which is a tree.
\end{proof}

\medskip

\begin{ex}
\exlab{octahedron}
Consider a regular dipyramid $P$ over a convex $2m+1$-gon;
see Fig.~\figref{DiPyramid}. 
One can see that, for every midpoint $x$ of a ``base edge'' $e$, 
$\C(x)$ is included in $\Sk(P)$.
More precisely, $\C(x)$ contains all base edges other than $e$, 
and the two ``lateral edges'' opposite to $x$.
In particular, this provides $2m+1$ such points, for $V=2m+3$ vertices.
\end{ex}

\begin{figure}[htbp]
\centering
\includegraphics[width=0.8\columnwidth]{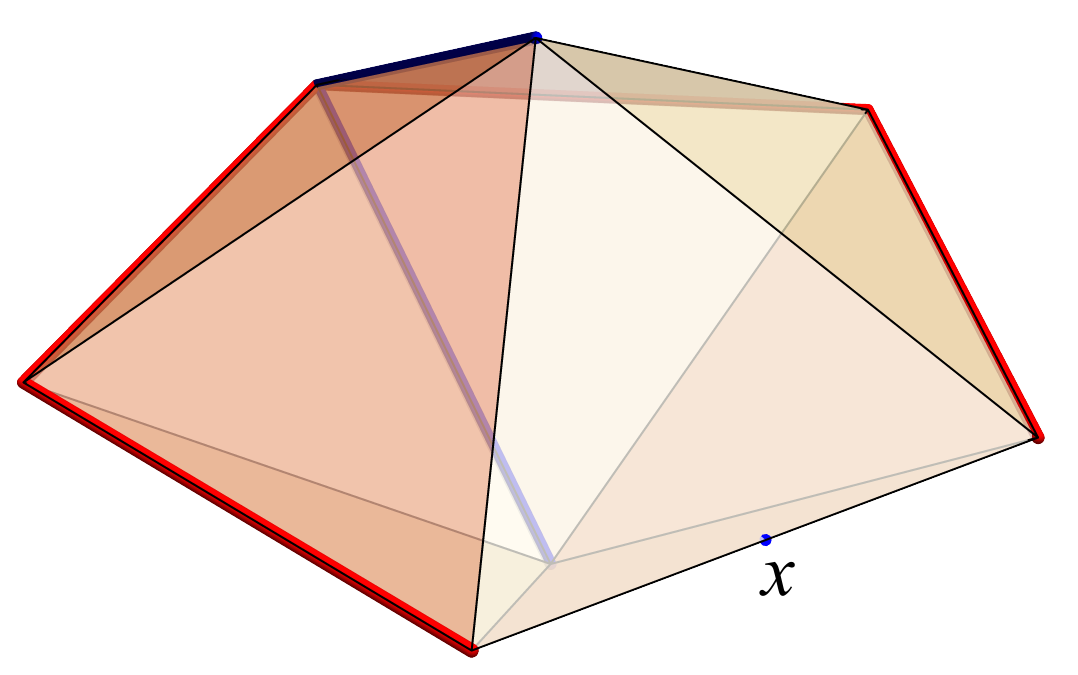}
\caption{$P$: pentagonal dipyramid. $\C(x)$: red and blue edges of $\Sk(P)$.}
\figlab{DiPyramid}
\end{figure}


\section{Absence of Skeletal Cut Loci}
\seclab{Absence of Skeletal Cut Loci}

\medskip
The following lemma will explain a condition in Theorem~\thmref{Rare} to follow.

\begin{lem}
\lemlab{EveryTetra}
Every tetrahedron $T$ has four points $x \in T$ such that 
$\C(x) \subset \Sk(T)$.
\end{lem}

\begin{proof}
For each vertex $v_i$, denote by $x_i$ the ramification point of $\C(v_i)$.
It follows, from cut locus property 
(ii), that that $v_i$ is the ramification point of $\C(x_i)$. 
Then, by (i) and (iii),
$\C(x_i)$ consists of the three edges incident to $v_i$.
\end{proof}

\medskip
The next theorem establishes the rarity of skeletal cut loci.
In the statement, by \emph{almost all} we mean ``all in an open and dense set''
in ${\mathfrak P}_V$.

\begin{thm}
\thmlab{almost-all-no-CL-Sk}
\thmlab{Rare}
For almost all convex polyhedra $P$ with $V >4$ vertices, there exists no point $x \in P$ with 
$\C(x) \subset \Sk(P)$.
\end{thm}
\noindent Note that Lemma~\lemref{EveryTetra} establishes the need for $V > 4$.

\medskip
\noindent
\begin{proof}
Notice first that almost all convex polyhedra $P$ are non-degenerate.

Assume, for the simplicity of the exposition, that every face of $P$ is a triangle and $\Sk(P)$ is a cubic graph.

\begin{description}
\item[Case 1.] 
Assume there exists a flat point $x$ interior to some face $F$ of $P$, such that $\C(x) \subset \Sk(P)$.

Repeating the notation in Theorem~\thmref{Finite},
denote by $v_i$, $i=1,2,3$, the vertices of $F$, and by $e_i$ the edges of $P$ 
incident to $v_i$ and not included in $F$.
Moreover, denote by $\g_i$ the geodesic segment from $x$ to $v_i$.

As in Theorem~\thmref{Finite},
it follows that $e_i \subset C(x)$ so, together, $\g_i$ and $e_i$ bisect the complete angle at $v_i$.
In other words, the straight extensions $E_i$ into $F$ by all the $e_i$
are concurrent: they all intersect at the same point.

Now we perturb the vertices of $P$ to destroy this concurrence.
If $P$ were a tetrahedron, then perturbing the apex would 
simultaneously move the edges incident to it.
But the assumption that $V > 4$ means that 
there are at least two vertices outside the $3$-vertex face $F$ containing $x$.
Perturbing these two vertices independently
moves the edges incident to $F$ independently, breaking the concurrence at $x$.

Because there are at most finitely many such points $x$ by Theorem~\thmref{Finite}, the conclusion follows in this case.

\item[Case 2.]  Assume there exists a flat point $x$ interior to some edge $e$ of $P$, 
such that $\C(x) \subset \Sk(P)$.
Denote by $v_i$, $i=1,2$, the vertices of $e$, and by $e_i$ the edges of $P$ 
incident to $v_i$ included in $\C(x)$.
As above, it follows that the straight extensions of $e_1,e_2$ coincide with $e$.
Now, small perturbations of the vertices of $P$ destroy this coincidence.
Note that if $e,e_1,e_2$  form a triangle, then $e_1,e_2$ will move together.
But still, perturbations at other vertices of $P$ (not $v_1,v_2, e_1\cap e_2$) will destroy the concurrence.

\item[Case 3.]
Assume finally there exists a vertex $v$ of $P$, such that $\C(v) \subset \Sk(P)$.
Here we obtain again that the straight extensions of two edges 
contain (other) edge-pair extensions, and small perturbations of the vertices of $P$ destroy this coincidence.
\end{description}
\end{proof}

\medskip

We mentioned the simple fact that, for the regular octahedron, 
for every vertex $v$, $\C(v)$ is skeletal.
In the next section we detail the special conditions such polyhedra must satisfy.


\section{Every Vertex a Skeletal Source}
\seclab{EveryVertex}
By Theorem~\thmref{Rare}, few convex polyhedra $P$ have a point $x$ with $\C(x) \subset \Sk(P)$.
So assuming that every vertex of $P$ has this property should yield some exceptional polyhedra.

\begin{thm}
\thmlab{Every-vertex}
Assume that every vertex of $P$ has a skeletal cut locus.
Then the following statements hold.
\begin{enumerate}
\item Every face of $P$ is a triangle.
\item Every vertex of $P$ has even degree in $\Sk(P)$.
\item The edges at every vertex $v$ split the complete angle at $v$ into evenly many sub-angles, every two opposite such angles being congruent.
\item If, moreover, every vertex of $P$ has degree $4$ in $\Sk(P)$ then $P$ is an octahedron: 
\begin{itemize}
\item with three planar symmetries, and
\item all faces of which are acute congruent (but not necessarily equilateral) triangles.
\end{itemize}
\end{enumerate}
\end{thm}
\begin{proof}
\begin{enumerate}[(1)]
\item Assume there exists a non-triangular face $F$ of $P$, so there are non-adjacent vertices $u,v$ of $F$.
Because $v \in C(u) \subset  \Sk(P)$, there exists an edge $vw$ of $P$ with $vw \subset C(u)$.
Moreover, the diagonal $uv$ of $F$ and $vw$ bisect the complete angle at $v$.

Because $vw$ is an edge, it is a geodesic segment from $w$ to $v$.
So $v$ is a leaf of $\C(w)$, and $\C(w)$ starts at $v$ in the direction of the diagonal $vu$, hence $\C(w) \not\subset \Sk(P)$.

\item Consider now a vertex $u$ of $P$ of degree $d$ in $\Sk(P)$, and
denote by $u_1,\ldots, u_d$ its neighbors in $\Sk(P)$.

For every $u_i$, $i=1,\ldots,d$, $u$ is a leaf of $\C(u_i)$, so the edge $u_i u$ and the edge of $\C(u_i) \cap \Sk(P)$ at $u$ bisect the complete angle at $u$.
Hence the edges at $u$ can be paired two-by-two, hence their number is even.

\item Denote by $e_1, \ldots, e_k, e_{k+1}, \dots, e_{2k}$ the edges sharing the vertex $u$, indexed circularly, 
and put $\a_i=\angle(e_i, e_{i+1})$, with index equality $2k+1=1$. 

The bisecting property of cut loci implies that
the edge $e_1$ (as a geodesic segment from vertex $u_1$ to $u$) 
and the edge $e_{k+1}$ (as the branch of $\C(u_1)$ at leaf $u$) 
bisect the complete angle at $u$:
$$\sum _{i=1}^k \a_i = \sum_{i=1}^k \a_{k+i}.$$

Similarly,
$$\sum _{i=2}^{k+1} \a_i = \sum_{i=2}^{k+1} \a_{k+i}.$$

Subtracting, we get $\a_1=\a_{k+1}$.

Analogous reasoning implies the other equalities: $\a_i=\a_{k+i}$, with index equality $2k+j=j$.

\item For the combinatorial part, 
denote by $F,E,V$ the number of faces, edges, and respectively vertices of $P$.
Euler's formula for convex polyhedra gives $F-E+V=2$.
Our assumptions imply $3F=2E$, and $4V=2E$.
These equations yield $V=6$ and $F=8$, hence $P$ is an octahedron.

Denote by $u,v,a,b,c,d$ the vertices of $P$, with $a,b,c,d$ neighbor to both $u$ and $v$.

Applying the hypothesis for $a,b,c,d$ shows that the cycle $C=abcda$ in $\Sk(P)$ is a bisecting polygon.
Therefore, there exists a local isometry $\I$ of the `upper' and `lower' neighborhoods $N_u, N_v$ of $C$.
In particular, the curvatures at $u$ and $v$ are equal, by Gauss-Bonnet.

It follows even more, that the local isometry $\I$ extends to an intrinsic isometry between the `upper' and the `lower' closed half-surfaces bounded by $C$ (regarding them as cones), 
hence it further extends to an isometry of $P$ fixing $C$.
Therefore, $C$ is planar and $P$ is symetric with respect to the respective plane, 
by the rigidity part of Alexandrov's Gluing Theorem.

Repeating the reasoning for other pairs of `opposite' vertices shows that all faces of $P$ are congruent triangles.

The four faces sharing the vertex $u$ have congruent angles at $u$, hence those angles are acute.
\end{enumerate}
\end{proof}


\begin{ex}
\exlab{SKCL-dipyramids}
Suitable dipyramids over convex $2 m$-gons,
similar to Example~\exref{octahedron}, provide non-octahedron polyhedra whose the cut loci of the vertices cover the $1$-skeleton.
\end{ex}


\section{A Combinatorial Restriction}
\seclab{Combinatorial}
Already mentioned in the Abstract,
at a first glance there seems to be very little relation between the cut locus and the $1$-skeleton, as the first one is an intrinsic geometry notion, and the second one specifies the combinatorics of $P$.
A background connection between the two notions can however be established in two steps:
Alexandrov's Gluing Theorem connects the intrinsic and the extrinsic geometry of $P$, while Steinitz's Theorem relates the combinatorics to the extrinsic geometry.

In this section we provide an easy combinatorial restriction to the existence of skeletal cut loci
for cubic graphs, complementing the first part of Theorem~\thmref{Every-vertex}.

\medskip

In the literature, a spanning tree without degree-$2$ nodes is called a
\emph{HIST}.\footnote{%
HIST abbreviates ``homeomorphically irreducible spanning tree.'' 
See, e.g., \cite{goedgebeur2024hist} and the references therein.}
So every spanning tree of a HIST-free graph has a degree-$2$ node.

\begin{thm}
\thmlab{Restriction}
A HIST-free cubic polyhedral graph cannot be realized with skeletal cut loci.
\end{thm}
\begin{proof}
Lemma 2.8 in~\cite{Reshaping} shows that, at a vertex $v$ of $P$ of degree $3$ in $\Sk(P)$, the sum of any two face angles is strictly larger than the third angle.
Therefore such a degree-$3$ $v$ in the graph cannot be a degree-$2$ node in a cut locus, because of Property~\ref{iv} of cut loci.
\end{proof}

\begin{op}
Theorem~\thmref{Restriction} provides a necessary condition for a cubic polyhedral graph
to be realizable with a skeletal cut locus. Is it also sufficient?
\end{op}

\begin{cor}
Among the Platonic solids, only the regular tetrahedron and the regular octahedron have skeletal cut loci.
\end{cor}

\begin{proof}
Notice first that all tetrahedra have skeletal cut loci (Lemma~\lemref{EveryTetra}), 
as does the regular octahedron, because
 it is a special case of Theorem~\thmref{Every-vertex} and Example~\exref{SKCL-dipyramids}).

One can check straightforwardly that the cube and the dodecahedron graphs are HIST-free, 
hence these polyhedra do not admit skeletal cut loci by Theorem~\thmref{Restriction}.

We clarify next the situation of the icosahedron $I$.
By the proof of Theorem~\thmref{Finite}, 
the only candidate source points $x$ are the centers of the faces.
Direct considerations show that such a cut locus is not completely included in $\Sk(I)$, see Fig.~\figref{Icosahedron}.
\end{proof}
%
\begin{figure}[htbp]
\centering
\includegraphics[width=0.7\columnwidth]{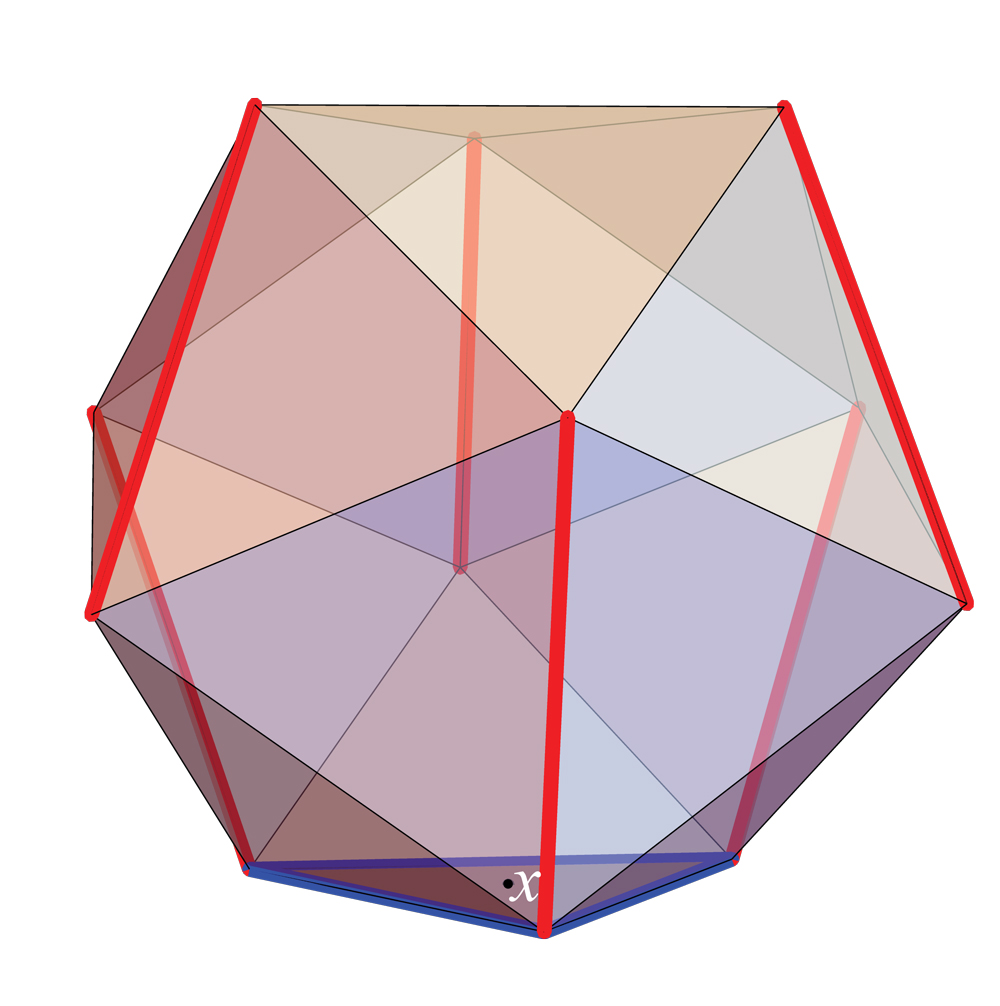}
\caption{Source $x$ is the center of the base face (blue).
The six disjoint red edges are included in $\C(x)$, but none of the other polyhedron edges are in $\C(x)$.
Therefore, $\C(x)$ must also include some non-edge geosegs to connect $\C(x)$ to a tree.
}
\figlab{Icosahedron}
\end{figure}

\medskip

The cube example shows there exist convex polyhedra admitting
edge-unfoldings, but not skeletal cut loci.


\section{Partially Skeletal Cut Loci}
\seclab{Partial-Edge Unfoldings}
Theorem~\thmref{Rare} states that 
almost all convex polyhedra $P$ with $V >4$ vertices have no skeletal cut loci.
This section is an attempt to roughly clarify ``how far'' are those polyhedra from having such cut loci.

\begin{rmk}
\rmklab{1-edge-unfoldings}
For every convex polyhedron $P$ and 
each edge $e$ of $P$,
there are infinitely many points $x \in P$ such that
$e \subset \C(x)$.
\end{rmk}

\begin{proof}
Consider an extremity $v$ of $e$, 
and the geodesic segment $\g_v$ starting at $v$ which,
with $e$, bisects the complete angle at $v$.
Also consider points $x$ on $\g_v$ sufficiently close to $v$.
Then, either the proof of Theorem~\thmref{Finite}, or
Property~\ref{iv} of cut loci in Section~\secref{Preliminaries} on which it is based,
directly implies $e \subset \C(x)$.
\end{proof}

\medskip

Fig.~\figref{Icosahedron} indicates a point $x$ on an icosahedron 
with $6$ polyhedron edges in $\C(x)$,
but several edges of $\C(x)$ are not part of the $1$-skeleton.

Fig.~\figref{OneShort} presents a cut locus which 
fails to be skeletal by a single edge.
In our example the respective edge is internal, but minor changes show that it could as well be external.
(An arc of a tree is called \emph{external} it it is incident to a leaf, and \emph{internal} otherwise.)
We believe that Theorem~\thmref{Rare} can be adapted to cover this case as well.
%
\begin{figure}[htbp]
\centering
\includegraphics[width=0.7\columnwidth]{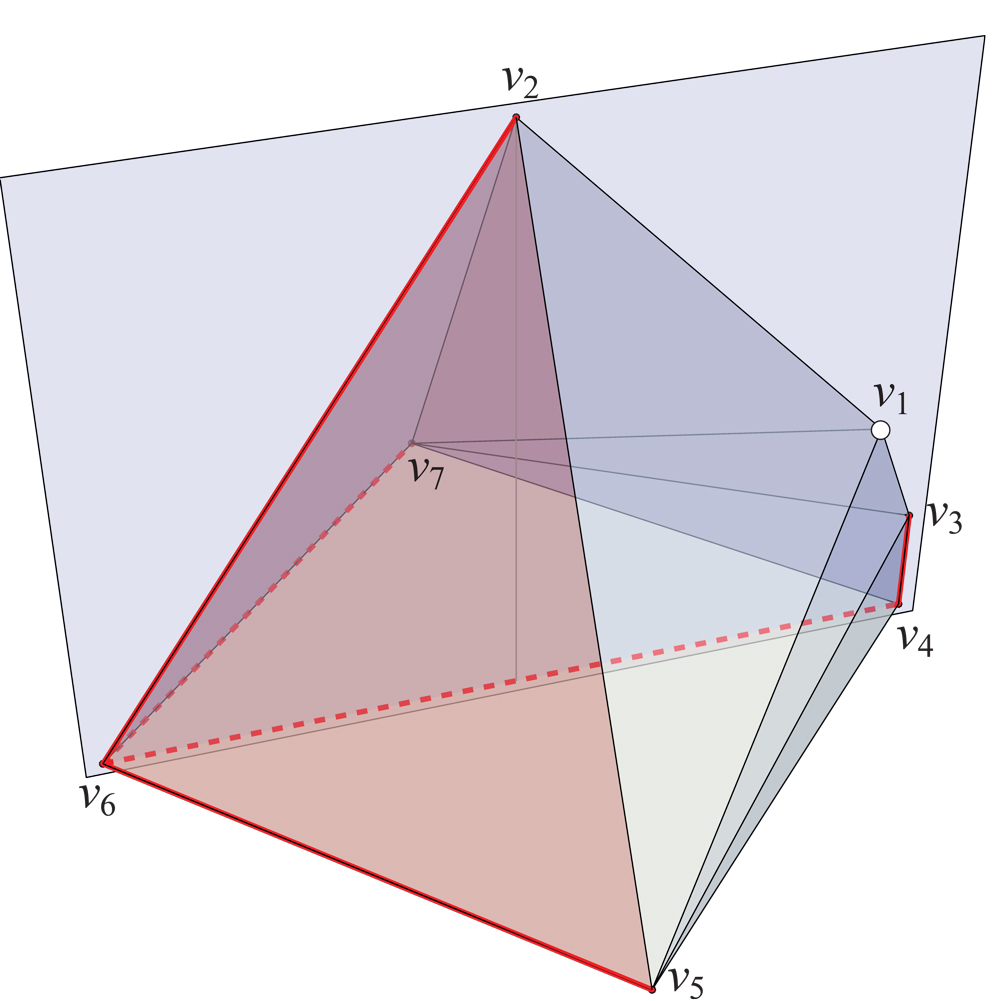}
\caption{A convex polyhedron of $7$ vertices and a vertical symmetry plane,
obtained from a regular pyramid $v_2 v_4 v_5 v_6 v_7$. 
The cut locus $\C(v_1)$, drawn red,
has $4$ polyhedral edges and one non-polyhedral edge, $v_4 v_6$.}
\figlab{OneShort}
\end{figure}

For a non-degenerate convex polyhedron $P$, define $L(P) \geq 1$ as 
the maximal integer such that
$P$ admits a point $x$ with $L(P)$ polyhedral edges in $\C(x)$.
The quantity $L(P)$ can be regarded as 
a measure of how close is $P$ to having skeletal a cut loci.

Direct considerations show that $L(C) = 4$ for the cube $C$, while $L(I) = 6$ for the 
icosahedron $I$ (see again Fig.~\figref{Icosahedron}).

The proof of Theorem~\thmref{EveryTree} can be easily adapted to provide examples of
convex polyhedra and cut loci missing an arbitrary number of their edges\footnote{
However, it doesn't work as such if one asks for the precise position of the non-polyhedral edges in the given tree.} 
from being skeletal.

\medskip

A \emph{$k$-edge-unfolding} is an unfolding whose cut tree contains precisely $k$ edges.

A \emph{partial edge-unfolding} is an unfolding whose cut tree contains at least one edge, 
but it is not skeletal;
so it is a $k$-edge-unfolding, for some $k \geq 1$.
These unfoldings are obtained by cutting partially along edges and partially outside edges,
so the concept is a bridge between 
edge-unfoldings (whose cut trees are composed by edges), and 
anycut-unfoldings (whose cut trees may contain no edge).

Remark~\rmkref{1-edge-unfoldings} shows that 
every convex polyhedron admits a $1$-edge-unfolding.

\medskip

The following two questions are now natural, 
and related to D\"urer's problem~\cite{o-dp-13}---whether or not 
every convex polyhedron has an edge-unfolding to a net.

\begin{op}
Find the maximal number $L_V \geq 1$ such that each non-degenerate convex polyhedron with $V \geq 5$ vertices\footnote{
Notice that $L_4=3$, by Lemma~\lemref{EveryTetra}.}
 has a point whose cut locus contains $L_V$ edges.

In particular, does every convex polyhedron have a point with two edges in its cut locus?
\end{op}

\begin{op}
Find the maximal number $K_V \geq 1$ such that each non-degenerate convex polyhedron with $V$ vertices has a $K(V)$ edge-unfolding to a net.
\end{op}

Clearly $1\leq L_V \leq K_V$, and Theorem~\thmref{Rare} shows that, for $V>4$, $L_V$ cannot equal the number of edges of a spanning tree with $n=V/2+1$ leaves
(see also the second paragraph in Section~\secref{Theorem1Conclusions}).

%

\medskip

The part of the paper presenting comments on, and consequences of, Theorem~\thmref{EveryTree} ends here.
The remaining is devoted to the proof of Theorem~\thmref{EveryTree},
which consists of a case analysis (Section~\secref{deg2}), a detailed construction for each (sub)case
(Sections~\secref{ConstructionDetails}--\secref{Examples}, 
\secref{Case_a},
\secref{Case_b},
\secref{Case_c},
\secref{Case_d}), 
and a concluding induction (Section~\secref{Induction Proof}).


\section{Proof of Theorem~\thmref{EveryTree}, Case of no Degree-$2$ Nodes}
\seclab{ConstructionDetails}
Throughout this section
we assume $\T$ has no degree-$2$ nodes.
Start with $P$ a pyramid with apex $a$ centered over a regular $n$-gon base $Q$, 
with $x$ the centroid of $Q$.
Label the vertices of $Q$ as $v_1,\ldots,v_n$.

The construction does not depend on the degree of apex $a$, so it is
no loss of generality to assume $a$ has degree-$3$ so that $P$ starts as a regular tetrahedron.
Let $z$ be a node of $\T$ adjacent to $a$. 
(We will often use $a$ and $z$ and
other variables to both refer to a node of $\T$ and a corresponding vertex of $P$.)
Let $z$ have degree $k+2$ in $\T$.
Truncation of $k$ planes through $z$ will create a vertex at $z$ of degree $k+2$.
E.g., if $z$ is degree-$3$, $k=1$ plane through $z$ creates a vertex of degree-$3$,
as we've seen in Fig.~\figref{Tfour_3d}.

We aim to understand how to truncate $k \ge 1$ planes through $z$
so that the $k+1$ truncation edges 
incident to the base $Q$ are part of $\C(x)$.
We will illustrate in detail the case $k=2$ shown in 
Fig.~\figref{k=2_nolab}.
Looking ahead, if we know how to construct $k$ planes through $z$,
then we can apply the same logic to construct $j$ planes through a child $y$ of $z$. 
The $j=1$ case is illustrated in
Fig.~\figref{k=2j=1_nolab}, with the red truncation triangle incident to $y$.
Then the same construction technique can be used to inductively
create the full subtree rooted at $z$.
We will show later that the subtrees rooted at the other two children of $a$
can be arranged to avoid interfering with one another.

\begin{figure}[htbp]
\centering
\includegraphics[width=0.97\columnwidth]{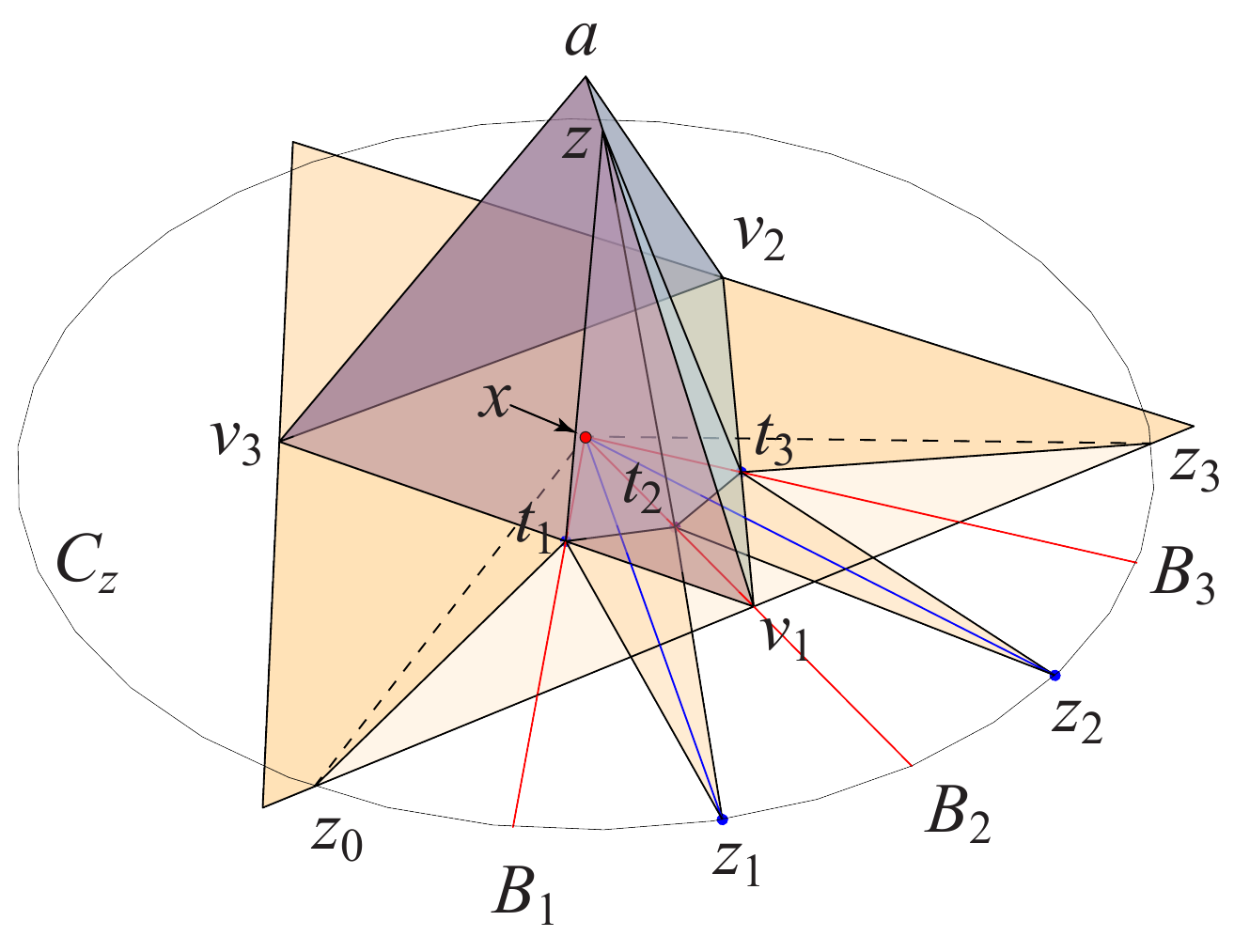}
\caption{$k=2$ truncation planes through $z$.}
\figlab{k=2_nolab}
\bigskip\bigskip 
\centering
\includegraphics[width=0.97\columnwidth]{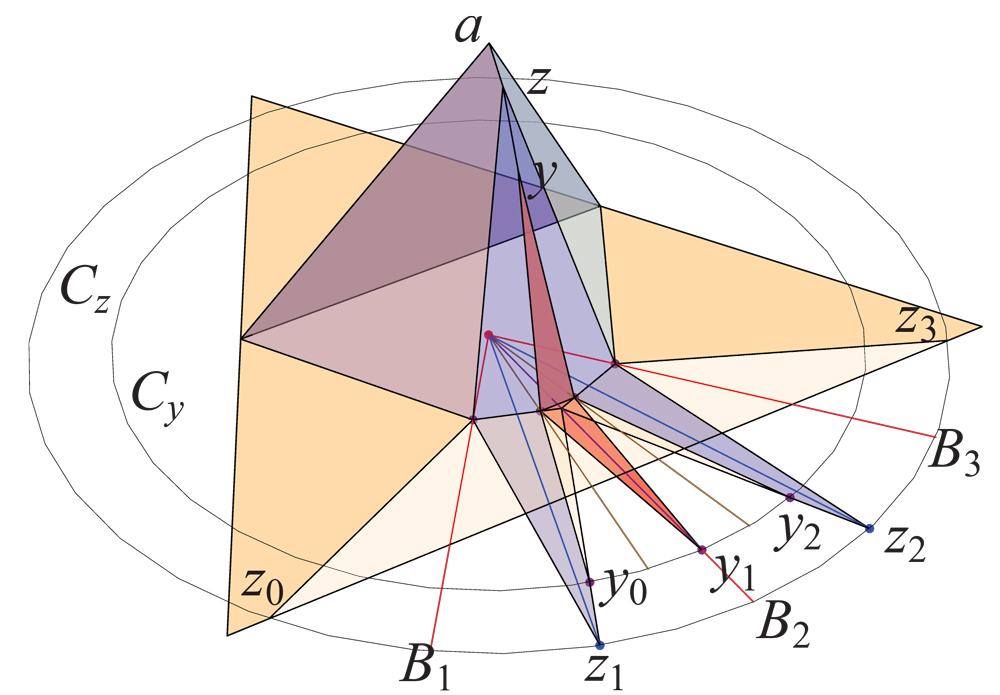}
\caption{$k=2$, $j=1$. The $y$-truncation cuts the $z t_2$ edge
in Fig.~\protect\figref{k=2_nolab}.}
\figlab{k=2j=1_nolab}
\end{figure}

We express the construction as a multi-step algorithm, and later prove that
the truncation edges are in $\C(x)$.
Fix $k \ge 1$, and position $z$ anywhere in the interior of $a v_1$.
The goal is to compute the \emph{truncation chain} 
$t_1,t_2,\ldots,t_k,t_{k+1}$ on base $Q$, where $t_1 \in v_1 v_n$ and $t_{k+1} \in v_1 v_2$
(e.g., $t_1, t_2, t_3$ in Fig.~\figref{k=2_nolab}). 
Each truncation triangle is then $z t_i t_{i+1}$.

The construction of the truncation chain is effected by first computingd
the unfolded positions $z_i$, the images of $z$ in the unfolding.
It is perhaps counterintuitive, but we can calculate $z_i$ without knowing 
 $t_i t_{i+1}$; instead we use $z_i$ to calculate $t_i t_{i+1}$.
The next construction depends of our choice of several parameters; we'll see later that it provides a suitable polyhedron.

\begin{enumerate}[(1)]
\item $z_0$ is the position of $z$ unfolded with the left face of the tetrahedron, $a v_3 v_1$.
$z_0$ can be determined by $| v_1 z | = | v_1 z_0 |$.
Then $z_{k+1}$ is the reflection of $z_0$ across $x v_1$.
\item Set $r_z = | x z_0 | = | x z_{k+1} |$.
\item All the $z_i$'s are chosen to lie on the circle $C_z$ centered on $x$ of radius $r_z$. 
\item Let $A$ be the angle $z_0 x z_{k+1}$.
Partition $A$ into $k+1$ angles $\a$.
This is another choice,
to maximize the symmetry of the construction.
\item The $z_i$'s lie on rays from $x$ separated by $\a$.
Together with $C_z$, this determines the location of the $z_i$'s.
\item Set $B_i$ to bisect the 
angle at $x$ between the
$z_{i-1}, z_i$ rays, $i=1,\ldots,k+1$.
%
\item We determine $t_1$ and $t_{k+1}$ using the first and last bisector:
$t_1 = v_1 v_n \cap B_1$, $t_{k+1} = v_1 v_2 \cap B_{k+1}$. 
The intermediate chain vertices $t_2,\ldots,t_k$ are not yet determined.
\item Let $\Pi_i$ be the mediator plane through $z z_i$,
the plane orthogonal to $z z_i$ through its midpoint.
It is these planes that determine $t_i$, $i=2,\ldots,k$.
%
\item $\Pi_i$ intersects the $xy$-plane in a line $L_i$ containing $t_i t_{i+1}$.
\item $t_i = L_i \cap B_i$.
\end{enumerate}

First note that the mediator plane construction of $t_i t_{i+1}$ guarantees that
$z$ unfolds to $z_i$.
Second, the angles between
edges $t_i z_{i-1}$ and $t_i z_i$ are split by $B_i$ by construction.
So any point $p$ on the interior of edge $z t_i$ unfolds to two images in the plane equidistant from $x$.

\begin{lem}
\lemlab{TruncEdges}
Each truncation edge $z t_i$ is an edge of $\C(x)$.
\end{lem}
\begin{proof}
We first prove that $z t_1$ lies in $\C(x)$.
Throughout refer to Fig.~\figref{Abstract_proof}.

Before truncation, the segment $z t_1$ lies on the face $a v_3 v_1$ 
of the polyhedron $P$, which is a regular tetrahedron in this case.

Fix a point $p \in z t_1$.
The unique shortest path $\g$ to $p$ crosses edge $v_1 v_3$.
After truncation, $\g$ remains a geodesic arc. We aim to prove that it remains shortest,
and moreover there is another companion geodesic segment $\g'$, establishing that $p \in \C(x)$.

Now we consider the situation after truncation. Let $\d$ be a geodesic arc from $x$ to $p$,
approaching $p$ from the other side of $z t_1$; see Fig.~\figref{Abstract_proof}(b).
If $\d$ crosses the edge $t_1 t_2$, then we have $| \g | = | \d |$ by construction,
and we have found $\g'=\d$. 

Suppose instead that $\d$ crosses edge $t_i t_{i+1}$ for $i \ge 2$, and then crosses
the truncation triangles
$z t_i t_{i+1}, z t_{i-1} t_i, \ldots , z t_1 t_2$
(right to left, i.e., clockwise, in Fig.~\figref{Abstract_proof}(a)) before reaching $p$.
To simplify the discussion, we illustrate $i=2$, so $\d$ crosses
$t_2 t_3$ and then triangles $z t_2 t_3$ and $z t_1 t_2$.
See Fig.~\figref{Abstract_proof}(b).


Let $q_2$ be the quasigeodesic\footnote{
A \emph{quasigeodesic} is a path with at most $\pi$ surface to either side of every point.} 
$x t_2 z$ on $P'$; it must be crossed by $\d$ to reach $p$. 
There are two triangles $x t_2 z_1$ and  $x t_2 z_2$
bounding $q_2$ to either side, congruent by the construction.
Thus the construction has local intrinsic symmetry about $q_2$.

\begin{figure}[htbp]
\centering
\includegraphics[width=0.80\columnwidth]{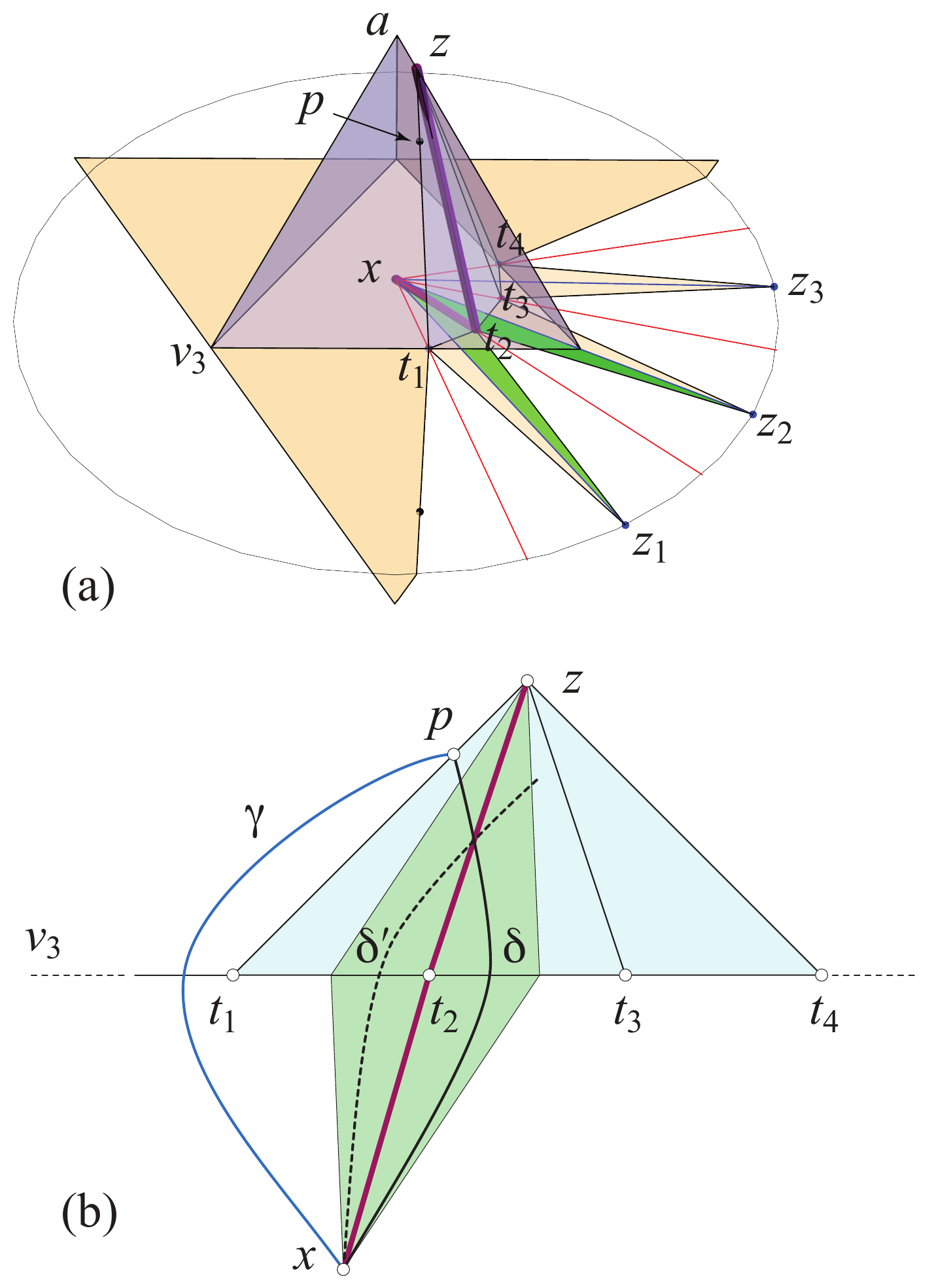}
\caption{Proof that $p \in z t_1$ is on $\C(x)$.
(a)~Quasigeodesic $q_2 = x t_2 z$ shown purple
and congruent triangles $x t_2 z_1$ and  $x t_2 z_2$ shaded green.
(b)~Abstract picture 
depicting geodesic segments $\g, \d, \d'$.}
\figlab{Abstract_proof}
\end{figure}

Let $s$ be the point at which $\d$ crosses $t_2 t_3$, $ \{s \} = \d \cap t_2 t_3$.
First assume that $s$ lies in the triangle $x t_2 z_2$.
Then $\d$ remains in $x t_2 z_2$ until it crosses $q_2$.
Then there must be another geodesic arc $\d'$ symmetric with $\d$ about $q_2$,
as illustrated in (b).
So $\d$ and $\d'$ meet at a point of $q_2$.
Because $\d$ and $\d'$ have the same length, neither can be a shortest path
beyond that point of intersection. Therefore $\d$ cannot reach $p$
as a geodesic segment.

Second, if $s$ instead lies in the triangle $x t_3 z_2$, then it is clear from the planar image in (a) of
the figure that $\d$ cannot cross the segment $x z_2$ clockwise, which it must to reach $p$
from the right in the figures. So $\d$ must head counterclockwise, 
crossing $q_3 = x t_3 z$. Then the same argument applies,
based this time on the local intrinsic symmetry about $q_3$, 
and shows that $\d$ cannot be a shortest path beyond $q_3$.

We have established that every point $p$ on $z t_1$ is on $\C(x)$, and
so $z t_1 \subset \C(x)$.
The same argument applies to $z t_{k+1}$, the rightmost truncation edge in the figures.

So now we know that two geodesic segments from $x$ to $z$ cross
$t_1 t_2$ and $t_k t_{k+1}$.
These two segments determine a digon $D$ within which the remaining segments
of $\C(x)$ lie.
But within $D$ we have local intrinsic symmetry with respect to the quasigeodesics
$q_i = x t_i z$, because $q_i$ is surrounded by the 
congruent triangles $x t_i z_{i-1}$ and $x t_i z_i$.
Therefore, the previous argument shows that all the edges $zt_i$ are included on $\C(x)$.
\end{proof}

\medskip

We now return to the claim that the three subtrees descendant from $a$ do
not interfere with one another.

\begin{lem}
\lemlab{NoInterference}
The truncations for one subtree descendant of apex $a$
do not interfere with another subtree descendant.
\end{lem}

\begin{proof}
First, as $k \to \infty$, $t_1$ approaches the line $x z_0$.
This is evident in Fig.~\figref{k=8_nolab} where $k=8$.
Thus the leftmost truncation triangle stays to the $v_1$-side of the midpoint of $v_1 v_3$,
say by $\e$.
Second, subsequent truncations to all but the extreme edges $z t_1$ and $z t_{k+1}$
stay inside the $t_1,\ldots,t_k$ chain.
The only concern would be that truncation of the $z t_1$ edge crossed 
the midpoint of $v_1 v_3$ (and so possibly interfering with truncations of $a v_3$).
However, as is evident in the earlier Fig.~\figref{Tfour_3d},
the position of $t_1$ moves monotonically toward $v_1$ as $z$ moves down $a v_1$.
Thus we can widen $\e$ to accommodate a truncation of $z t_1$ (or of $z t_{k+1}$).
So the entire subtree rooted at $z$ stays between the midpoints of $v_1 v_3$ and $v_1 v_2$.
\end{proof}

\medskip
\noindent
Further examples are shown in Section~\secref{Examples}:
$k=4$ in Figs.~\figref{k=4_nolab} and~\figref{k=4j=3_nolab},
and
$k=8$ in Figs.~\figref{k=8_nolab} and~\figref{k=8j=1_nolab}.

\medskip
Lemmas~\lemref{TruncEdges} and~\lemref{NoInterference} together establish this case of Theorem~\thmref{EveryTree}:
$\C(x) \subset \Sk(P)$ matches the given $\T$.

\bigskip

In this section we proved Theorem~\thmref{EveryTree} 
for trees $\T$ without degree-$2$ nodes.
Our construction can be viewed as realizing degree-$2$ nodes 
of $\T$ with flat ``vertices'' on $\Sk(P)$---points interior to edges of $P$.
The passage from flat vertices to positive curvature vertices is a long proof,\footnote{
In this respect, there is some similarity to the proof of Steinitz's Theorem.}
accomplished in the following,
after giving a few more examples for the current construction in the next section.


\section{Further Examples}
\seclab{Examples}

In this section we illustrate the previous construction with more examples.

\begin{figure}[htbp]
\centering
\includegraphics[width=0.8\columnwidth]{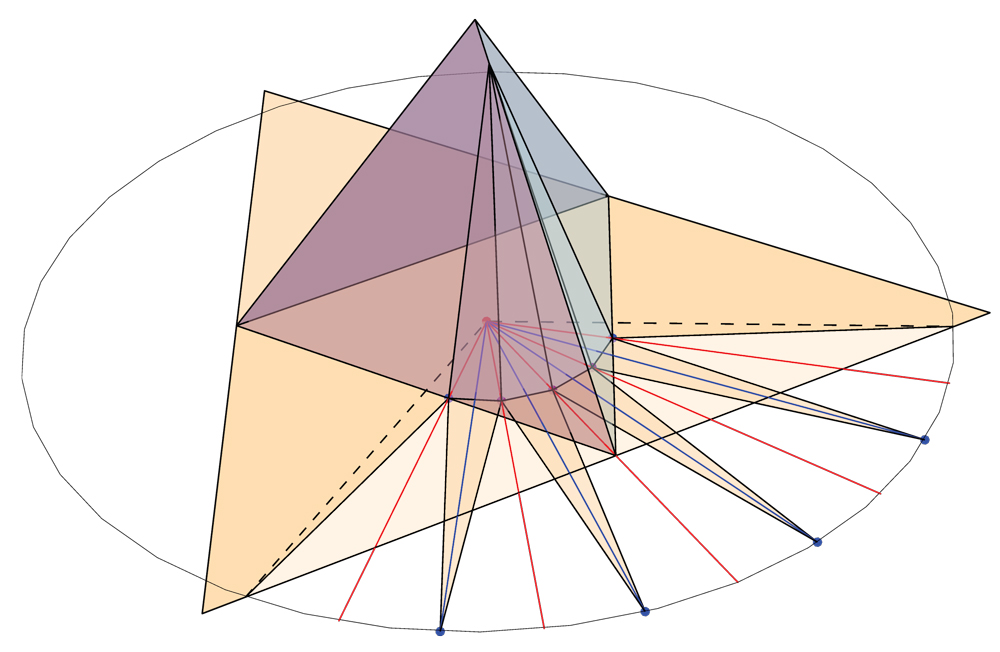}
\caption{$k=4$.}
\figlab{k=4_nolab}
\bigskip
\centering
\includegraphics[width=0.8\columnwidth]{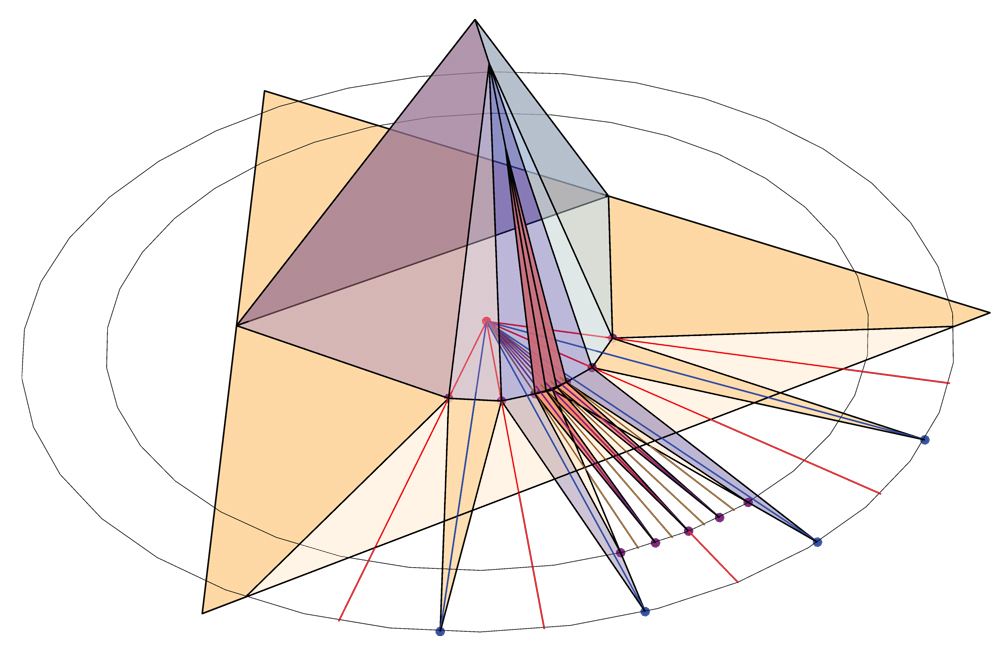}
\caption{$k=4$, $j=3$.}
\figlab{k=4j=3_nolab}
\end{figure}

\newpage 
\begin{figure}[htbp]
\centering
\includegraphics[width=1.0\columnwidth]{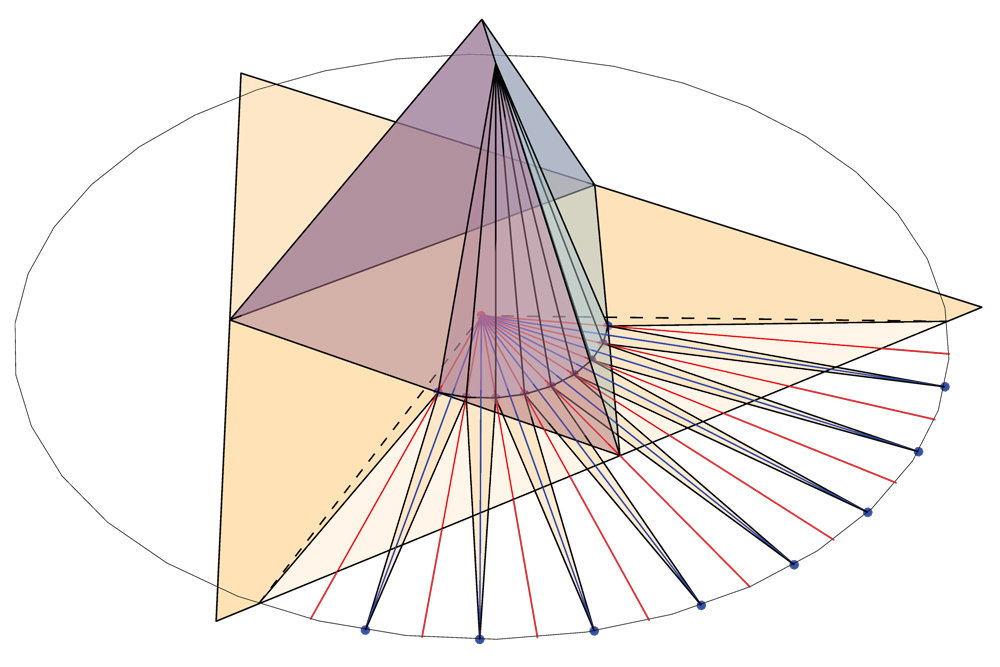}
\caption{$k=8$.}
\figlab{k=8_nolab}
\bigskip\bigskip 
\centering
\includegraphics[width=1.0\columnwidth]{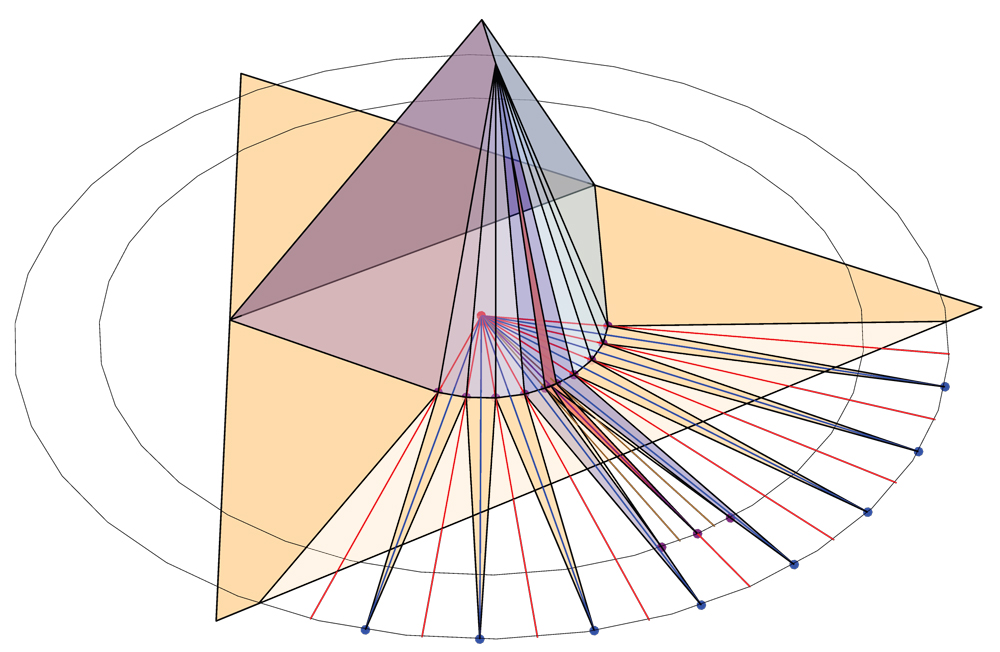}
\caption{$k=8$, $j=1$.}
\figlab{k=8j=1_nolab}
\end{figure}

\clearpage
\section{Degree-$2$ Nodes: Four Cases}
\seclab{deg2}

We turn now to degree-$2$ nodes.
The overall plan is to start with a zero-curvature degree-$2$ node $u$ identified
on an edge $ab$ of $\C(x)$.
Then conveniently bend the edge at $u$ by moving $b$
so that $u$ gains positive curvature, 
while maintaining that $\C(x)$ includes $a u$ and $u b$.

The bending at $u$
introduces two new polyhedron edges incident to $u$ on each side.
Those two new edges could both end on the base,
or one terminating on the base and the other on a non-base vertex,
or both edges terminating on a non-base vertices.
See Fig.~\figref{FourCases} for examples of each case,
and Fig.~\figref{StackedPyramids} for a polyhedron falling in Case~(d).
Each case will be further described in the appropriate section.

\begin{figure}[htbp]
\centering
\includegraphics[width=1.0\textwidth]{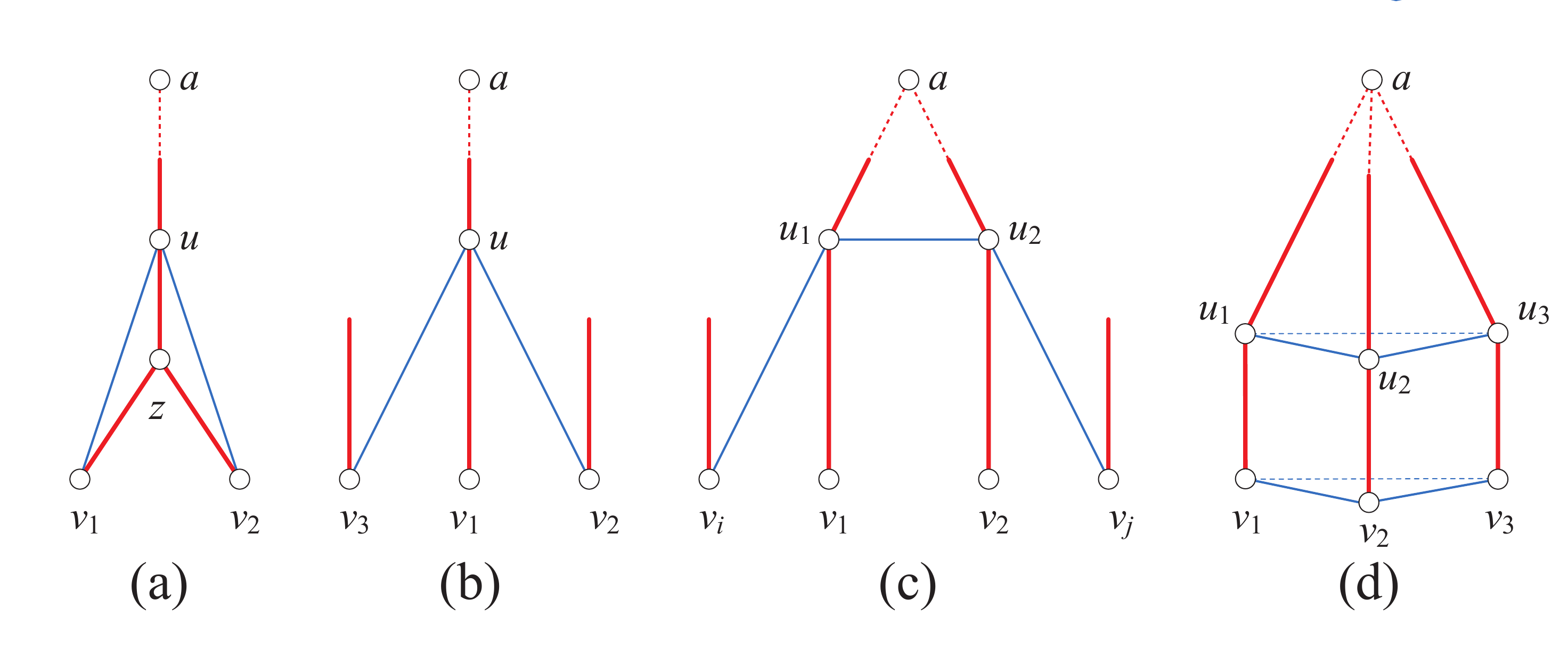}
\caption{Examples of four cases. Red: $\C(x)$ edges. Blue: edges of $P$.
}
\figlab{FourCases}
\end{figure}

\begin{figure}[htbp]
\centering
\includegraphics[width=0.95\linewidth]{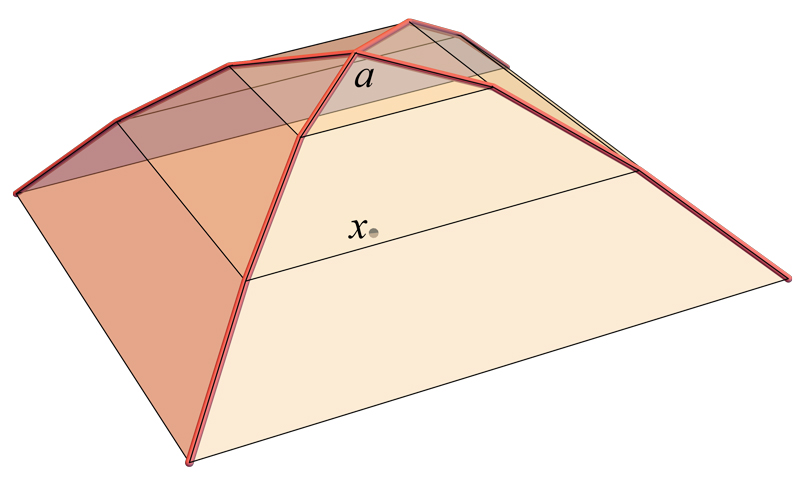}
\caption{Case~(d): Cycles of degree-$2$ nodes.}
\figlab{StackedPyramids}
\end{figure}


\subsection{Case~(a)}
\seclab{Case_a}
In both Case~(a) and Case~(b), a degree-$2$ vertex $u$ is connected on both sides to
base vertices $v_i, v_j$.
Case~(a) occurs when $u$ is a parent of a non-base vertex $z$, whereas in Case~(b),
$u$ is a parent of a base vertex. See
Fig.~\figref{FourCases}(a,b).

In Case~(a), the degree-$2$ vertex $u$ can realized by modifying the construction
that achieves Lemma~\lemref{TruncEdges}.
It will suffice to show how to deal with a degree-$2$ node $u$ a child of apex $a$ in the tree $\T$, and $z$ a child of $u$ of degree $\ge 3$.
The construction generalizes to arbitrary placements of such degree-$2$ nodes.

So let $u$ be on edge $a v_1$ but $z$ on edge $u v'_1$, where $v'_1 \neq v_1$ is on the line segment $x v_1$.
See Fig.~\figref{Deg2_k1}.
Thus $u$ is a degree-$4$ vertex of $P$, but we want to arrange that two of its edges
are not part of $\C(x)$.
The two segments $a u$ and $u z$ are in $\C(x)$, as they lie on the vertical 
symmetry plane containing $a x v_1$.

\begin{figure}[htbp]
\centering
\includegraphics[width=1.0\textwidth]{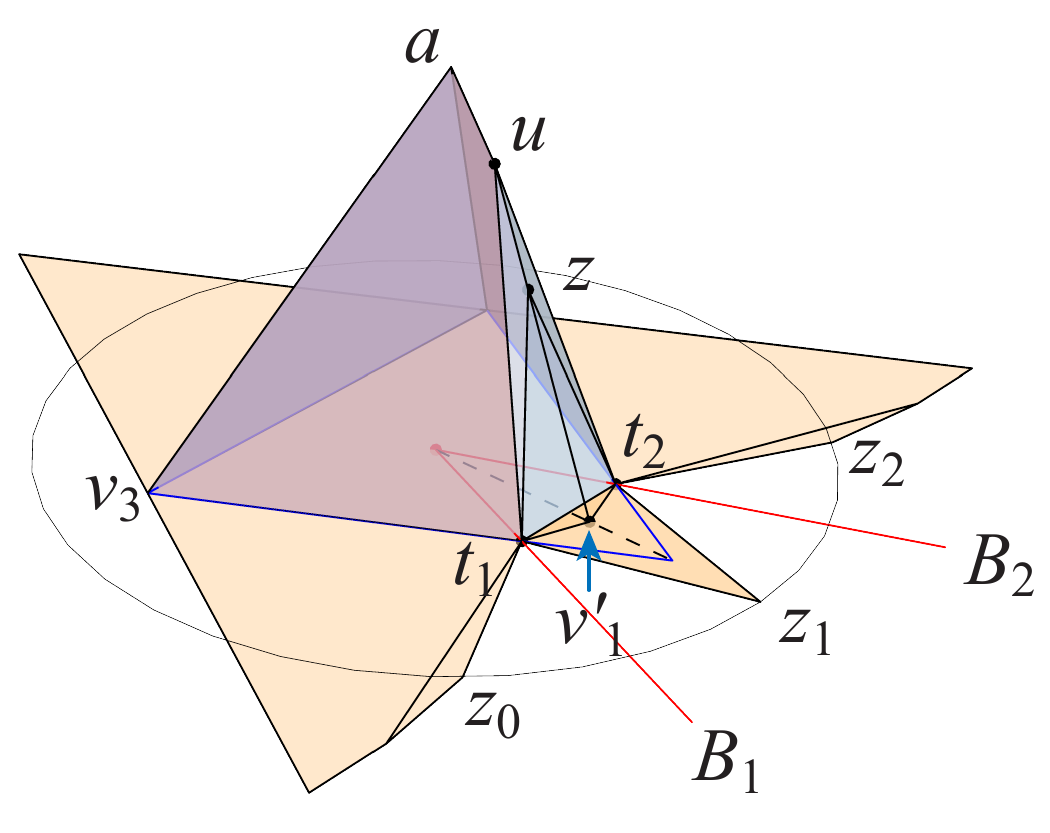}
\caption{Case~(a). $u$ is degree-$2$ node. $k=1$ truncation at $z$.}
\figlab{Deg2_k1}
\end{figure}

Note that the triangle $u z t_1$ is not coplanar with the left face $v_3 t_1 u a$.
Still, when we truncate through $z$, then cut the truncation edges and unfold,
that triangle $u z t_1$ unfolds attached to the unfolding of the left face.
We perform the same calculations to truncate $k$ times at $z$, and the same logic
(bisectors $B_i$ and mediator planes $\Pi_i$) leads to the conclusion that the
truncation edges are part of $\C(x)$.

The two side edges $u t_1$ and $u t_2$ are not part of $\C(x)$:
a point $p \in u t_1$ is closer to $x$ via a geodesic segment up the left face, 
closer than any other path from $x$ to $p$.
So $u$ has degree-$4$ in $\Sk(P)$ but degree-$2$ in $\C(x)$.


\subsection{Case~(b)}
\seclab{Case_b}
In this case, a degree-$2$ node $u$ is a parent of a leaf node $v_1$,
as illustrated in Fig.~\figref{FourCases}(b). 
This requires the two non-$\C(x)$ edges incident to $u$ to connect
to adjacent base vertices $v_2$ and $v_3$, and consequently to other branches of $\C(x)$. 

We confine ourselves now to starting with a regular tetrahedron 
with $x$ at the centroid of the base and $a$ the apex, as in Case~(a).
We first describe the construction at a high level, concentrating on $v_3$,
with the understanding $v_2$ will be handled similarly.

For the regular tetrahedron, the surface angle incident to $v_3$
can be partitioned into two halves,
to each side of the path $q = x v_3 \cup v_3 a$,
$\a$ to the left of $q$, with
$\a = \angle x v_3 v_2 + \angle v_2 v_3 a$,
and $\b$ to the right of $q$, with 
$\b = \angle x v_3 v_1 + \angle v_3 v_3 a$.
Because $\a=\b$, the edge $v_3 a \subset \C(x)$.
Introducing $u$ on the path $a u v_1$ 
modifies the angle at $v_3$ to the right, to $\b' < \b$, breaking the bisection.
So our goal is to temporarily increase $\b'$ by altering the path $a u v_1$,
so that by further truncations we could decrease it to $\a$.

In particular, $u$ will connect to $v''_1 \neq v_1$.
(See ahead to Fig.~\figref{Caseb}.)
The vertices $v_2,v_3,a$ and the source $x$ remain unaltered,
which ensures the changes are localized, and that the construction
generalizes beyond the regular tetrahedron.

Let $v'_1$ be a point on the ray $x v_1$, beyond but close to $v_1$.
(For now, this slightly increases $\b'$.)
We'll place $u$ on the segment $a v_1'$.
But first we locate an auxiliary variable point $y \in a v'_1$
which will bound $u$ to lie above $y$ on the segment $y a$.
The construction will work for any $u$ in that range.

Let $y_1$ be the vertical projection of $y$ onto $x v'_1$.
Now consider the polyhedron $P = \conv \{v_2,v_3, a, y, y_1 \}$.
We examine the angle $\d$ incident to $v_3$ from the right of the
path $q$. Angle $\d$ will serve as $\b'$.
It is composed of three angles: on the base, the face including vertical
edge $y \, y_1$, and the face including $a$:
\begin{align*}
\d = \angle x v_3 y_1 + \angle y_1 v_3 y + \angle y v_3 a \;.
\end{align*}
See Fig.~\figref{DeltaPoyh}.
\begin{figure}[htbp]
\centering
\includegraphics[width=0.75\textwidth]{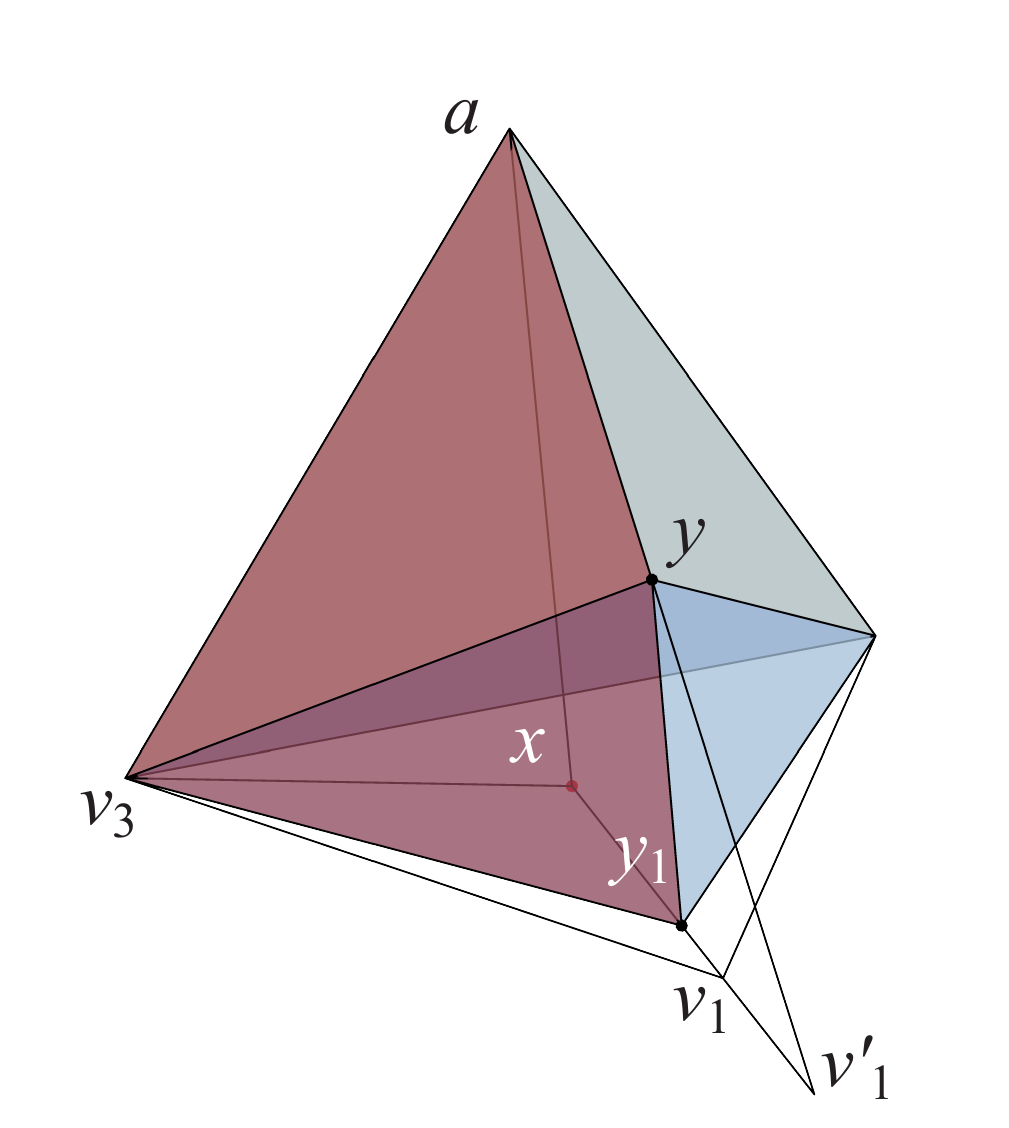}
\caption{Case~(b). 
The total angle at $v_3$ is 
$\angle x v_3 y_1 + \angle y_1 v_3 y + \angle y v_3 a$.
}
\figlab{DeltaPoyh}
\end{figure}

Next view $y=y(t)$ as continuously varying on $a v'_1$, with
$y(0) = a$ and $y(1) = v'_1$.
So $\d=\d(t)$ is also varying continuously.
We now argue that the extreme values of $\d$ are less than and respectively greater than
the fixed angle $\a$.

\begin{itemize}
\item At $t=0$, $\d(0) = \angle x v_3 a$ because the first two angle terms
are zero when $y \, y_1 = a x$.
This angle $\d(0)$ is smaller than $\a$ because 
it is smaller than each of the two angles comprising $\a$. 
\item At $t=0$, $\d(1) = \angle x v_3 v'_1 + \angle v'_1 v_3 a$,
because $y = y_1 = v'_1$.
This angle $\d(0)$ is larger than $\a$ by
the triangle inequality for spherical distances
(see e.g., Lemma~2.8 in~\cite{Reshaping}).
\end{itemize}

Therefore, there is some $t$ such that $\d(t) = \a$.
We henceforth define $y$ to be that $y(t)$.
The resulting polyhedron $P=P(t)$ achieves $\a = \b$.
Thus if we place $u$ at $y$, we have achieved our goal of increasing $\b'$ to $\b$,
and thus ensuring that $v_3 a \in \C(x)$.
However, for more complicated trees, we may need to truncate the 
vertical edge $y \, y_1$,
so we would like choose $u$ to slant to $v''_1$.

We now claim that there is some $v''_1$ on $x v'_1$ such
that the three angles to the right of $q$ sum to exactly $\a$.
The analogous three angles are as in the $\d$ argument above;
see Fig.~\figref{Caseb}.

\begin{align*}
\angle x v_3 y_1 + \angle y_1 v_3 u + \angle u v_3 a < \a \;.
\end{align*}
%
\begin{align*}
\angle x v_3 v'_1 + \angle v'_1 v_3 u + \angle u v_3 a > \a \;.
\end{align*}
%
Therefore there is some $v''_1 \in y_1 v'_1$ so that
the construction using the slanted edge $u v''_1$ achieves $\a$,
guaranteeing that $\a = \b$ and $v_3 a \in \C(x)$.
(Note that $v''_1$ could be closer to $x$ than $v_1$, or further.)

As in Case~(a),
the two segments $a u$ and $u v_1''$ are in $\C(x)$, 
as they lie on the vertical symmetry plane containing $a x v_1''$.

\begin{figure}[htbp]
\centering
\includegraphics[width=0.75\textwidth]{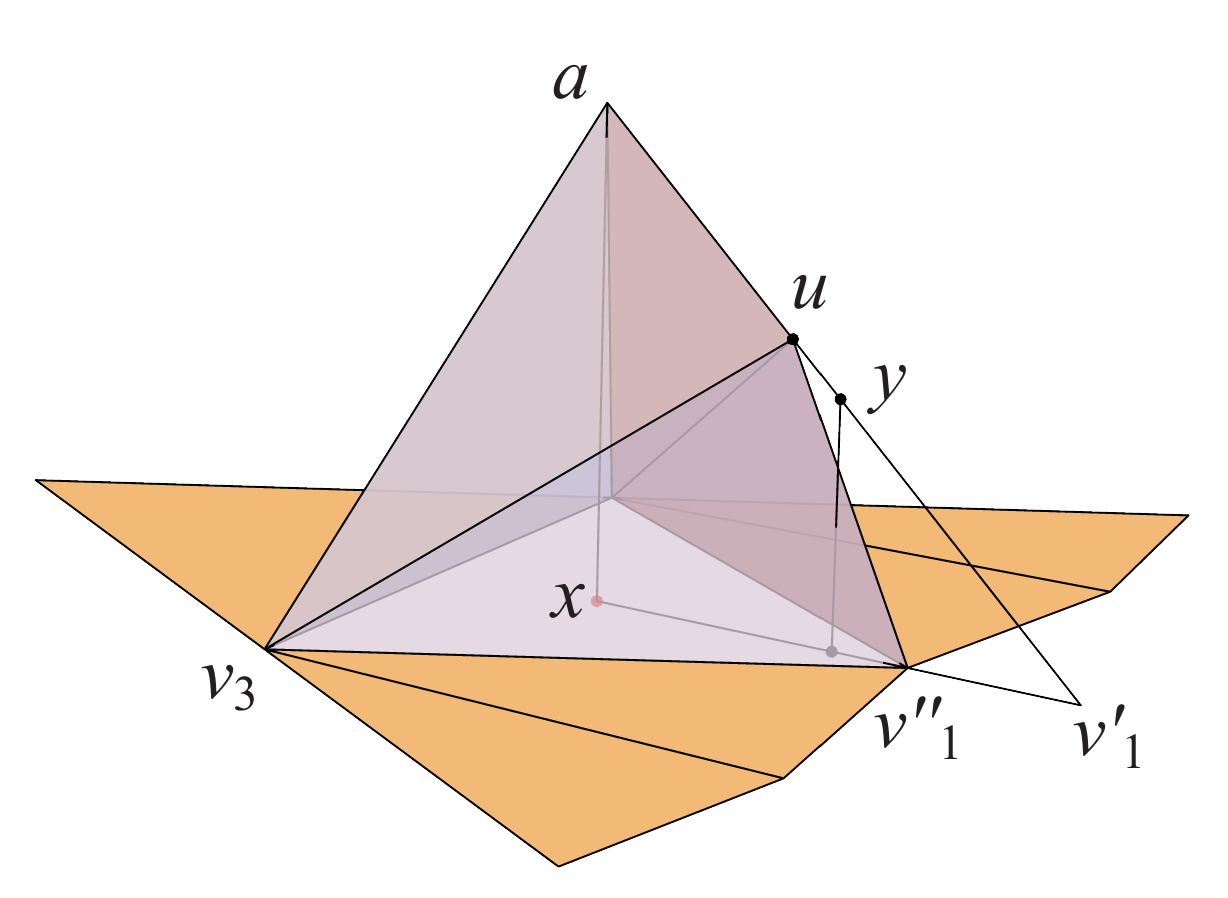}
\caption{Case~(b): Source unfolding after locating $v''_1 \in y_1 v'_1$.}
\figlab{Caseb}
\end{figure}

We note here that, if there are several nodes of degree-$2$ 
falling under Case~(b)
in a cascade, they should be treated 
together by a similar procedure.

\clearpage
\subsection{Case~(c)}
\seclab{Case_c}

It is characteristic of Case~(c) that a chain $C$ of degree-$2$ nodes, 
represented by the single edge $u_1 u_2$ in Fig.~\figref{FourCases}(c), 
connects to base vertices $v_i, v_j$, $1\leq i<j \leq n$, on either end of $C$.
In contrast, in Case~(d) the chain closes to itself, as in
Fig.~\figref{StackedPyramids}.


The proof in Case~(c) is quite involved and comprises several steps.
Some of those steps describe the construction, while the others give the necessary argument for the construction to work.
The proof falls roughly into two parts,
the distinction of which will be useful for Remark~\rmkref{About Case (c)} at the end of the proof.
\textbf{Part~I} (Item~\ref{Item 1} to Item~\ref{dist-x-a'-a13}) finds the points $a_{13}, u_{13}$, 
and then \textbf{Part~II} (Item~\ref{a-12} to Item~\ref{a=b}) finds the points $a_{12}, u_{12}, y_1, u_{21}$.
These points will be defined as they occur in the proof.


\medskip
\noindent
\textbf{Part~I}.
\begin{enumerate}[(1)]
\item 
\label{Item 1}
Start with the regular tetrahedron\footnote{
The proof works, with minor changes, for arbitrary regular pyramids.} 
$S= av_1v_2v_3$, and let $x$ be the projection of $a$ on the base plane $\Pi_b = v_1v_2v_3$.

Take $v_3' \in [xv_3]$.

This has two purposes. One is to destroy the chance of obtaining in the end as a solution precisely the original tetrahedron $S$.
The other one is to assure the convexity of the resulting polyhedron $R$, and will become apparent at Item~\ref{convexity}.

Denote by $T$ the tetrahedron $a v_1 v_2 v_3'$.

\item Choose $u_1 \in [av_1]$ and take $u_2 \in [av_2]$ such that $|au_1|=|au_2|$.

%

Fig.~\figref{Proof_u12_3} illustrates the setup so far.
The vertices $\{ v_3', a, u_1, u_2 \}$ will henceforth remain fixed.
The plan is to move $v_1 v_2$ parallel to itself toward $x$, to $y_1 y_2$.
This will introduce a bend at $u_1$ and $u_2$. The remainder of the argument aims
to calculate $y_1 y_2$ so that $\C(x)$ includes $a u_1 y_1$ and $a u_2 y_2$.

It may help to look ahead to the final construction:
see Fig.~\figref{Proof_u12_12}.
In some sense, we are constructing the source unfolding of the 
``to-be-found'' polyhedron.

%

Because of the symmetry of $T$, we will concentrate on the $u_1$ side of $T$.

The triangle $v_3' a u_1$ of $T$ is now fixed;
it will become a face of the final polyhedron.


In the source unfolding, it will become $v_3' a_{13} u_{13}$.
A key is locating this unfolded face, i.e., determining the position of $a_{13}$.

\begin{figure}[htbp]
\centering
\includegraphics[width=1.0\textwidth]{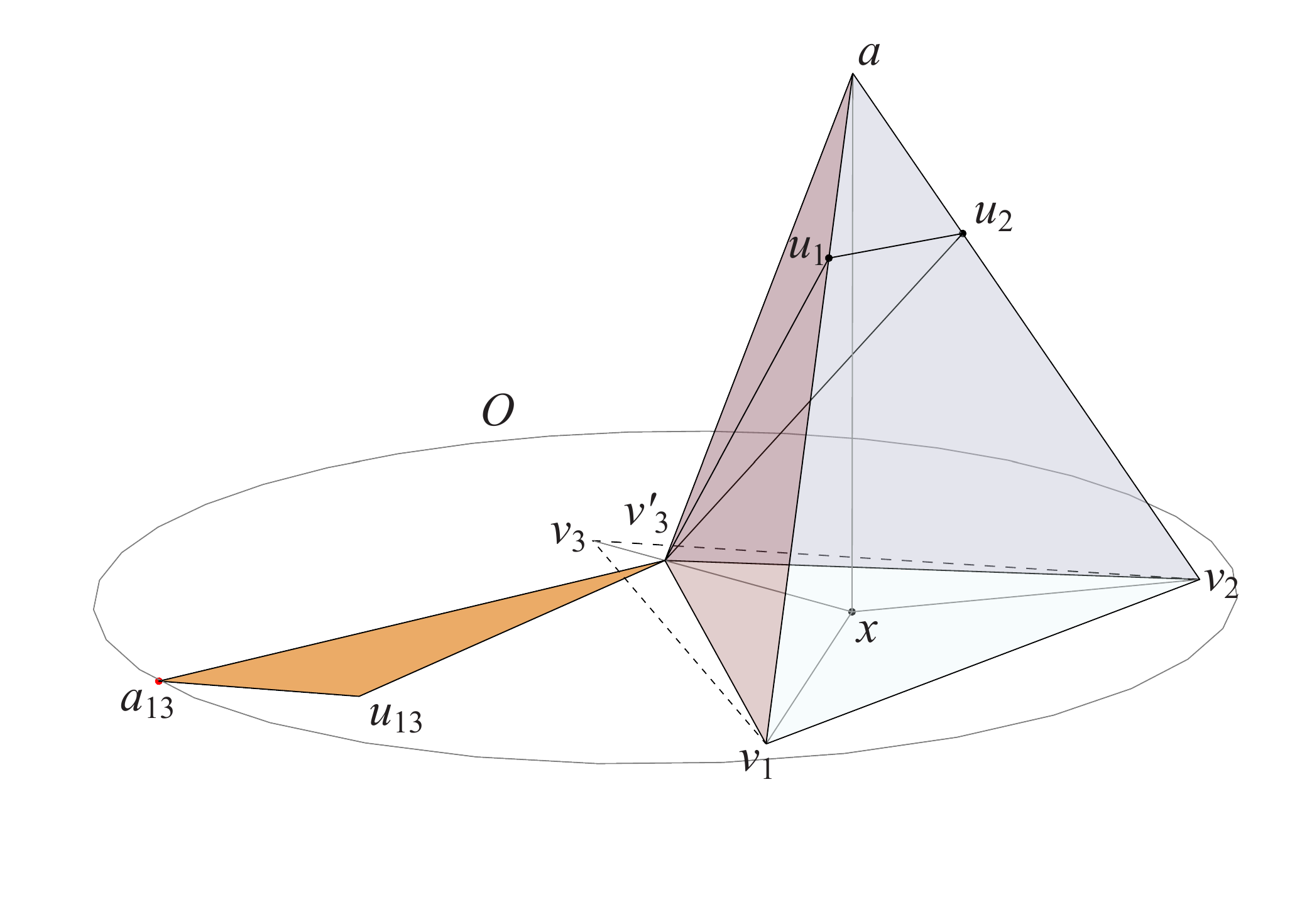}
\caption{$v_3$ has been moved to $v_3'$. 
The unfolding of face $v_3' a u_1$ is shown.}
\figlab{Proof_u12_3}
\end{figure}

%
\item
\label{phi*}
Let $\q_{13},\q_{12}$ be the angles at the apex $a$:

$$\q_{13} = \angle u_1av_3' \; \; {\rm and}\; \; \q_{12} = \angle u_1 a u_2.$$

A key angle will be $\q_{13}-\q_{12}/2$.
Assuming the construction is finished and looking
ahead, the argument will partition $\q_{13}$ as follows;
see Fig.~\figref{Proof_u12_6}.
%
\begin{figure}[htbp]
\centering
\includegraphics[width=1.0\textwidth]{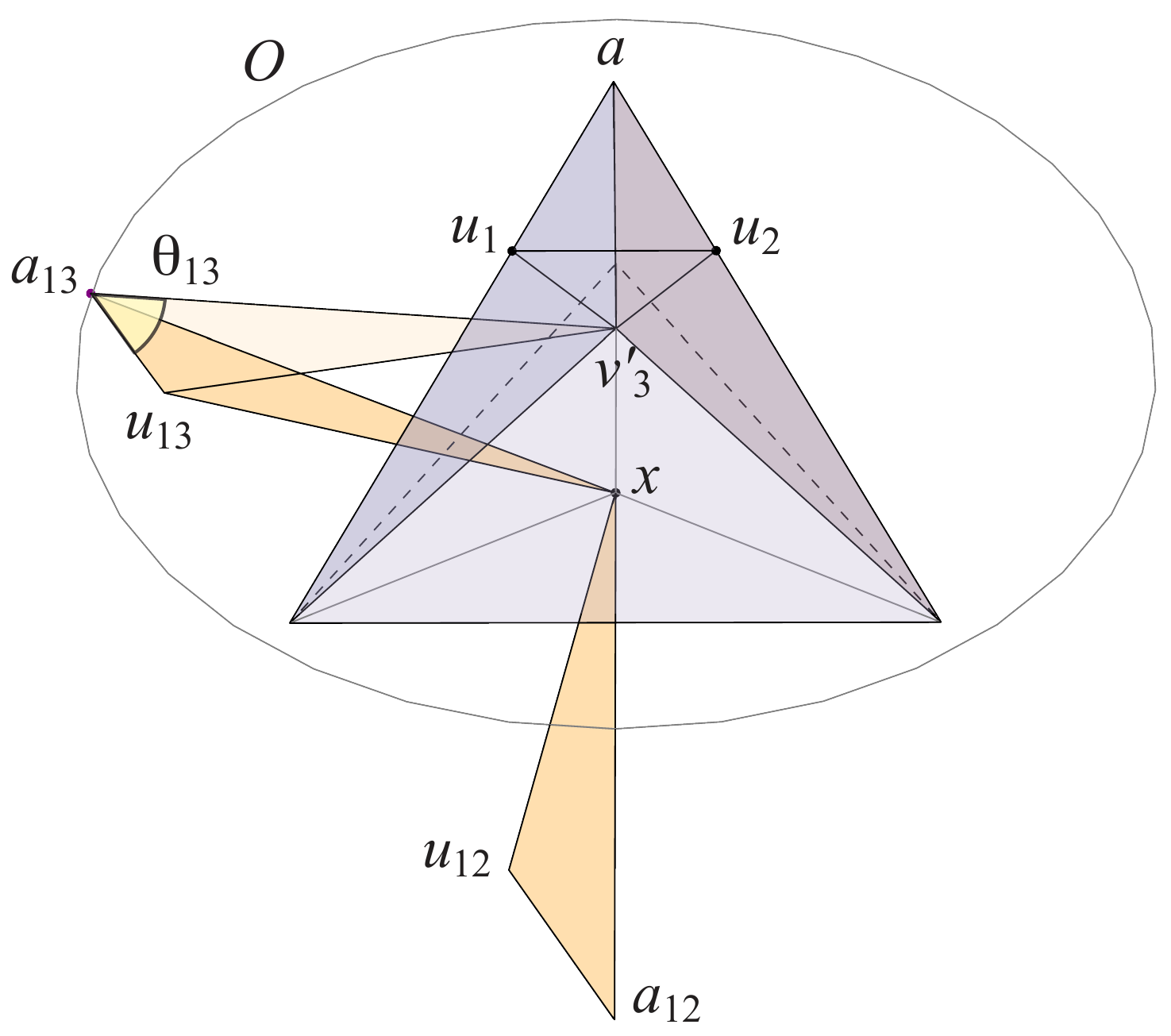}
\caption{$\q_{13} - \q_{12}/2$.
The labeled points $a,u_1,u_2$ are
above the base plane, and all the other labeled points lie in the base plane.
}
\figlab{Proof_u12_6}
\end{figure}
%
\begin{align*}
\q_{13} & = \angle u_1 a v_3' \\
& = \angle u_{13} a_{13} v_3' \\
& = \angle x a_{13} v_3' + \angle x a_{13} u_{13} \\
& = \angle x a_{13} v_3' + \angle x a_{12} u_{12} \\
& = \angle x a_{13} v_3'  + \q_{12}/2 \\
\q_{13} - \q_{12}/2 & = \angle x a_{13} v_3' 
\end{align*}


\item 
Notice that the source unfolding of the target polyhedron $R$ would produce a point $a_{13}$ (an image of $a$) in $\Pi_b$ with $|a_{13} v_3'|=|av_3'|$.
So we consider in $\Pi_b$ the circle $O$ centered at $v_3'$ and radius $|av_3'|$, and will determine $a_{13} \in O$ by Items~\ref{min-max-p} and~\ref{a13}.

Notice that $x$ is inside $O$.

\item
\label{min-max-p}
\begin{lem}
\lemlab{mi-max-p}
Consider a variable point $z \in O$.
The extreme values of $\p= \p (z) =\angle xzv_3'$ are as follows.

The minimum value of $\p$ is $0$, achieved precisely for $z,x,v_3'$ collinear.

The maximum $\p_0$ of $\p$ is obtained precisely for $zx$ perpendicular to $xv_3'$, hence for two positions of $z$.

Moreover, on each of the four arcs of $O$ determined by those extreme values, 
$\p (z) =\angle xzv_3'$ is a strictly monotone function on $z$.
\end{lem}

\begin{proof}
The ``minimum'' claim is clear: there are precisely two positions of $z \in O$ for which the mimimum value of $\p$ is attained, say at the north and the south poles of $O$.

To prove the rest, take the height from $v_3'$ in the triangle $v_3'xz$, and notice that it is smaller than, or equal to, $[xv_3']$.
Therefore, because $\sin$ is an increasing function, we get the two positions $z_0 \in O$ of $z$ obtaining maximal value $\p_0$ of $\p$: $z_0 x \perp xv_3'$.
One of them is in the left semi-circle of $O$ and another one in the right half-circle.

The monotonicity of the angular function $\p(z)$ on each of the resulting four arcs is elementary.
\end{proof}

\noindent
We will denote $z_0$ by $a_0$.

\item 
\label{Four-a13}
Lemma~\lemref{mi-max-p} helps to prove the following.

\begin{lem}
There exist precisely four positions $z_{13}$ of the variable point $z \in O$ s.t. 
$$\angle x z_{13} v_3' = \q_{13} - \q_{12}/2.$$
\end{lem}


We described above in Item~\ref{phi*}
why the angle $\q_{13} - \q_{12}/2$
is the appropriate choice to determine the position of $a_{13}$ on $O$.

\begin{proof}
Notice that the variable $\p$ defined in Lemma~\lemref{mi-max-p}
clearly depends continuously on $z$, and 
\begin{equation}
\label{q13-q12/2}
0 < \q_{13} - \q_{12}/2 < \p_0,
\end{equation}
where $0$ is the minimal value of $\p$, and $\p_0$ is its maximal value given by Lemma~\lemref{mi-max-p}.

The first inequality follows for the regular tetrahedron from the numerical values
$\q_{13} \approx \pi/3$ (for $v_3'$ close to $v_3$) and $\q_{12} = \pi/3$.\footnote{
If, instead of the regular tetrahedron, we would start with an arbitrary regular pyramid, we would use Lm.2.8 in \cite{Reshaping} to derive the conclusion.}

To see that $\p_0 > \q_{13} - \q_{12}/2$, consider $a_0 \in O$ such that $a_0 x \perp xv_3'$.
The triangles $a_0xv_3'$ and $axv_3'$ are congruent, because $x$ is the projection of $a$ on $\Pi_b$.

Next we show that 
$$\angle x a v_3 > \angle  v_1 a v_3 / 2 =\angle  v_1 a v_3 - \q_{12}/2.$$
To see this, let $xp \perp v_3 v_1$, with $p \in v_3 v_1$. The theorem of the three perpendiculars implies $ap \perp v_3 v_1$, hence 
$$\sin (\angle  v_1 a v_3 / 2) = \sin \angle pav_3 = |pv_3|/|av_3| < |xv_3| / |av_3| 
= \sin \angle xav_3.$$
For $v_3'$ close enough to $v_3$ we have $\angle xav_3' \approx \angle xav_3$, so we still have\footnote{
This argument is valid for arbitrary regular pyramids.}
$$\p_0 = \angle xav_3' > \angle  v_1 a v_3 / 2 = \angle  v_1 a v_3 - \q_{12}/2 > \q_{13} - \q_{12}/2.$$

By the continuity of $\p$ and the inequalities (\ref{q13-q12/2}), there are precisely four intermediate positions of $z \in O$ for which $\angle xz_{13}v_3' = \q_{13} - \q_{12}/2$.
In the left semi-circle of $O$, one is below and one is above the point $z_0$ realizing the maximum of $\p$ (i.e., $z_0 \in O$, $z_0 x \perp xv_3'$). 
\end{proof}

\item
\label{a13}
Fig.~\figref{Proof_u12_5} shows the position of $z_{13}$ we choose for (i.e., denote by) $a_{13}$: the left one above $z_0$.
This choice will be important at Item \ref{a'}.

\item 
\label{congruent-triangles}
Construct the triangle $v_3' a_{13} u_{13}$ congruent to $v_3'au_1$,
with $u_{13}$ inside the angle $\angle x v_3' a_{13}$.

\item
\label{a'}
Rotate the triangle $v_3' a v_1$ about $v_3' v_1$ until it lies in the base plane, and denote by $a'$ the resulting image of $a$ farthest from $v_2$, so $a'$ clearly lies in the left half-circle of $O$.

It follows that $a'$ is above $a_{13}$; i.e., $a_{13}$ lies between $a_0$ and $a'$, on the left half of $O$.
This is a direct consequence of
Items~\ref{Four-a13},~\ref{a13},~\ref{comparison}, and Lemma~\lemref{comp}.

\begin{figure}[htbp]
\centering
\includegraphics[width=1.0\textwidth]{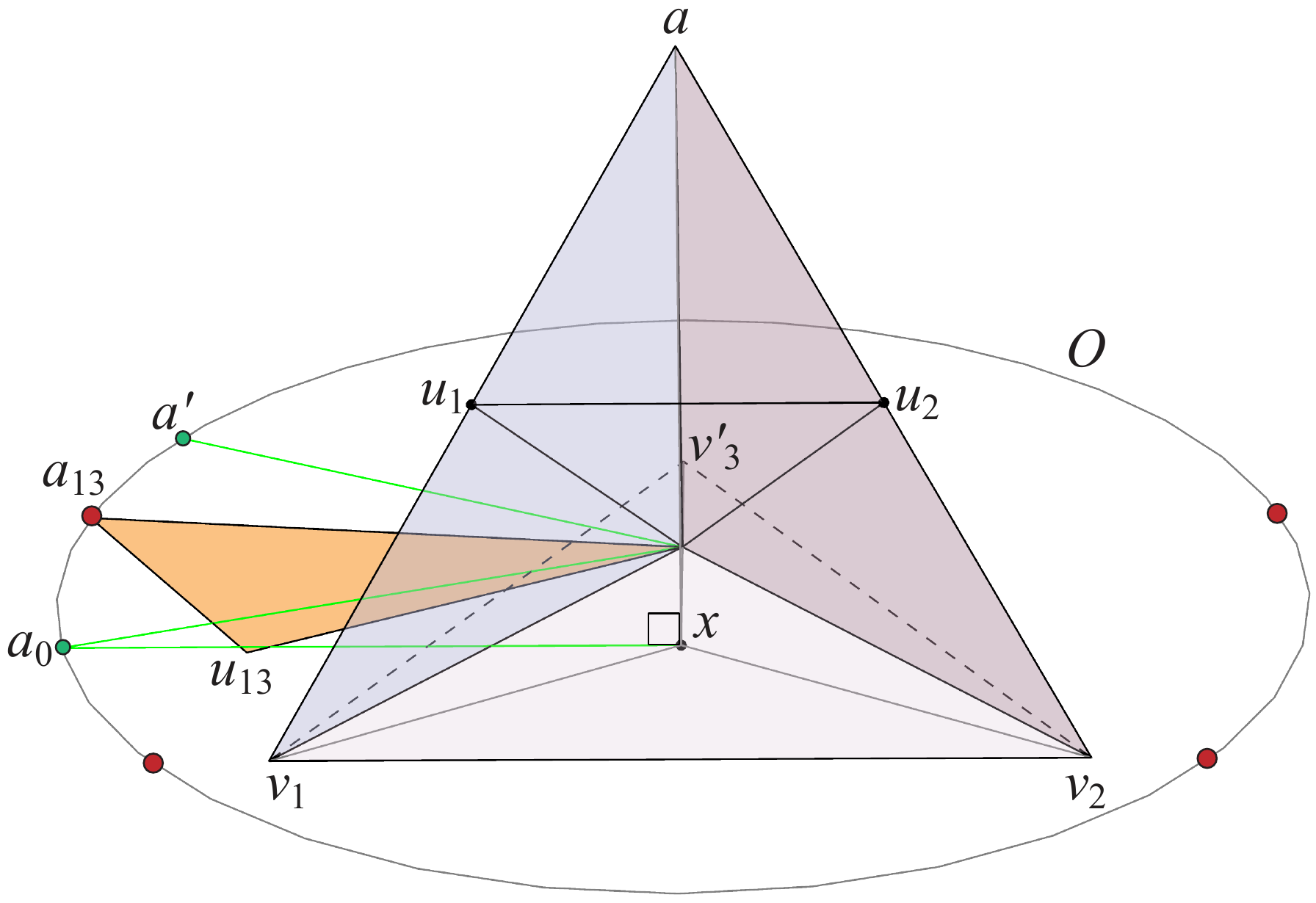}
\caption{The four $\p$ solutions (red), and $a_{13}$ 
between (green) points $a_0$ below and $a'$ above.
The labeled points $a,u_1,u_2$ are 
above the base plane, and all the other labeled points lie in the base plane.
}
\figlab{Proof_u12_5}
\end{figure}

\item
\label{comparison}
\begin{lem}
\lemlab{comp}
$\angle v_3' a' x < \q_{13} - \q_{12}/2$.
\end{lem}

\begin{proof}
Notice that the claimed inequality is equivalent to
$\q_{12}/2 < \q_{13} - \angle v_3' a x = \angle x a v_1$.

Let $m \in v_1 v_3'$ such that $xm \perp v_1 v_3'$.
See Fig.~\figref{mmp}.

\begin{figure}[htbp]
\centering
\includegraphics[width=0.75\textwidth]{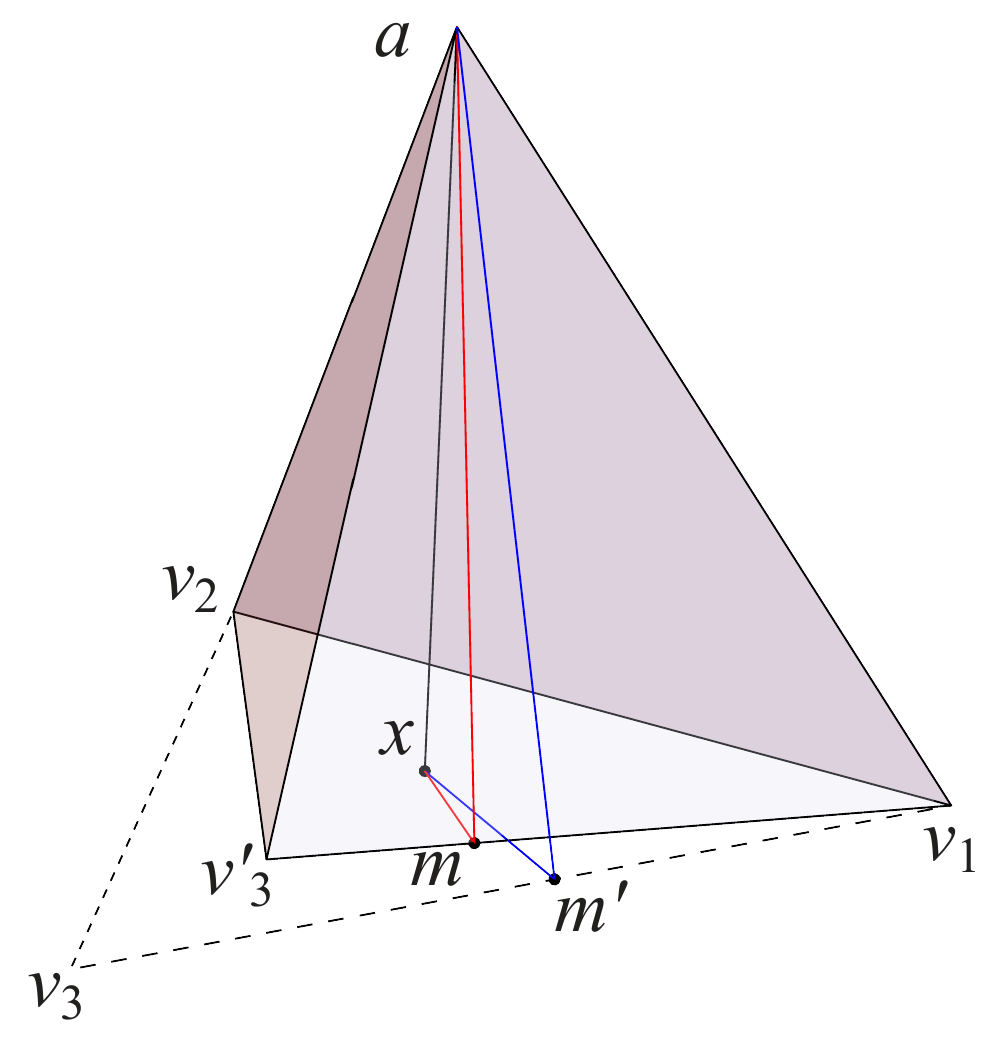}
\caption{$xm \perp v_1 v_3'$ and $xm' \perp v_1 v_3$.
The geoseg $x m a$ (red) on $T$ is shorter than $x m' a$ (blue) on $S$.
}
\figlab{mmp}
\end{figure}

By the theorem of the three perpendiculars, $am \perp v_1 v_3'$, so the path $xma$ is a geodesic segment.
Therefore, $\angle x a v_1 > \angle m a v_1$.

In the tetrahedron $S=a v_1 v_2 v_3$, let $m' \in v_1 v_3$ such that $xm' \perp v_1 v_3$.
Then, again by the theorem of the three perpendiculars, $am' \perp v_1 v_3$, 
so the path $xm'a$ is a geodesic segment on $S$.
Therefore, $\q_{12}/2 = \angle v_3 a v_1 /2 = \angle m' a v_1$.

So the claimed inequality is 
reduced to 
$\angle m a v_1 > \angle m' a v_1$.
This follows from $\angle m v_1 a < \angle m' v_1 a$, which is a direct consequence of the choice of $v_3'$ between $x$ and $v_3$.
\end{proof}

\item
\label{dist-x-a'-a13}
The above argument also shows that $|xm| < |xm'|$ and $|ma| < |m'a|$, hence
$|xm| + |ma| < |xm'| +  |m'a|$. 
I.e., the distance from $x$ to $a$ is shorter on $T$ than on $S$.

Moreover, because $a_{13}$ lies between $a_0$ and $a'$, on the left half of $A$,
$|x a_{13}| < |xa'|$.
Therefore, $|x a_{13}|$ is shorter than the distance from $x$ to $a$ on $S$.

\noindent
This concludes identifying the points $a_{13}, u_{13}$.
\end{enumerate}

\noindent
\textbf{Part~II}
\begin{enumerate}[resume*] 
\item 
\label{a-12}
Now we begin the second part of the proof,
identifying the points $a_{12}, u_{12}, y_1, u_{21}$.

On the ray at $x$ orthogonal to $v_1v_2$, take the point $a_{12}$ determined by 
$|xa_{13}|=|xa_{12}|$.

\item 
\label{u12}
Construct, inside $\angle a_{12} x a_{13}$, the triangle
$a_{12} x u_{12}$ congruent to $a_{13} x u_{13}$.
We illustrated this step earlier in Fig.~\figref{Proof_u12_6}.

Clearly, the triangle $a_{13} x u_{13}$ lies inside $\angle v_3xv_1$, 
hence the triangle $a_{12} x u_{12}$ lies inside $\angle v_2xv_1$.
Therefore, the triangles $a_{12}xu_{12}$ and $a_{13} x u_{13}$ share only the point $x$. 

\item 
\label{L13}
The mediator plane $\Pi_{13}$ of $u_1, u_{13}$
intersects the plane $\Pi_b$ under the line $L_{13}$ through $v_3'$.

The point $y_1$ we are constructing should be at 
equal distances from $u_1, u_{13}$, and from $u_{12}$.
That $|y_1 u_1|=|y_1 u_{13}|$ is achieved by $\Pi_{13}$,
and that these distances equal $|y_1 u_{12}$ is achieved by $\Pi_{12}$.

Consider $u' \in a' v_1$ such that $|a' u'|=|a u_1|$.
Moving continuously $v_3$ to $v_3'$ would move continuously several objects: 
$a'$ to $a_{13}$;
$u'$ to $u_{13}$; 
the mediator plane $\Pi_1'$ of $u_1, u'$ to $\Pi_{13}$; 
and the intersection line $v_3' v_1 = \Pi_1' \cap \Pi_b$ to $L_{13}$.

Therefore, because $a_{13}$ lies between $a_0$ and $a'$ (see Item~\ref{a'}),
$L_{13}$ enters at $v_3'$ the triangle $v_3' v_1 v_2$ and 
consequently it intersects the edge $v_1 v_2$.

\item 
\label{y1}
The mediator plane $\Pi_{12}$ of $u_1, u_{12}$
intersects the plane $\Pi_b$ under the line $L_{12}$.
Because of Item~\ref{dist-x-a'-a13}, $L_{12}$ separates $x$ from $v_1 v_2$.

Denote by $y_1$ the intersection point of $L_{13}$ and $L_{12}$.

Then, by Item~\ref{L13}, $y_1$ lies inside the triangle $v_1 v_1 v_3'$.

\item It follows that $\angle u_1 v_3' y_1 = \angle u_{13} v_3' y_1$, hence the triangle $y_1 v_3' u_{13}$ folds to the face $y_1 v_3' u_1$.

\item Furthermore, by Item \ref{congruent-triangles}, the triangle $v_3' u_{13} a_{13}$ folds to the face $v_3' u_1 a$.

\item 
\label{symmetry}
Proceed similarly for $u_2$, to obtain the points $a_{23} \in O$, $u_{23}$, $u_{21}$ and $y_2$.
The construction is symmetric with respect to the plane $a x v_3 ' \perp v_1 v_2$.

\item 
\label{choice a13}
Notice that the triangles $au_1u_2$ and $a_{12} u_{12} u_{21}$ are congruent,
as they have congruent sides from $a$ resp. $a_{12}$, 
$a u_1 \equiv a_{13}u_{13} \equiv a_{12}u_{12}$ (by Items~\ref{congruent-triangles} and~\ref{u12}),
and congruent angles between those sides (by the choice of $a_{13}$ at Item~\ref{a13} and the consecquent costruction).
So $|u_1u_2|=|u_{12}u_{21}|$.

\item
\label{isosceles-trapezoids}
Because we have $|u_{12} y_1|=|u_1 y_1|=|u_{13} y_1|$ (see Items~\ref{L13}-\ref{y1}) and similarly for $u_{21}$, and by $|u_1u_2|=|u_{12}u_{21}|$,
the isosceles trapezoids $u_1 u_2 y_2 y_1$ and $u_{12} u_{21} y_2 y_1$ are congruent.

Therefore, the pentagon $a_{12} u_{12} y_1 y_2 u_{21}$ folds to the faces $a u_1 u_2$ and $u_1 u_2 y_2 y_1$.

Moreover, because $u_1 u_2 \parallel v_1v_2$, $L_{12}$ is also parallel to $v_1v_2$.

\item Because $|y_1 u_{13}|=|y_1 u_1|=|y_1 u_{12}|$, we have $\angle x y_1 u_{13} = \angle x y_1 u_{12}$.
Therefore, the edge $y_1 u_1$ is in $\C(x)$, by the bisecting property of $\C(x)$ at $y_1$.

By construction, we have that 
$\angle a_{13} u_{13} y_1 = \angle a_{12} u_{12} y_1$, hence the edge $a u_1$ is also in $\C(x)$.

\item 
\label{convexity}
Finally, notice that the resulting polyhedron is convex, because 
$y_1$ and $y_2$ lie inside the base $v_1v_2v_3'$ (Items~\ref{y1} and~\ref{symmetry}).

\item
\label{a=b}
We next argue that the angles incident to $u_1$ from either side---$\a$ to the left
and $\b$ to the right---are equal, proving that $\C(x)$ bisects at $u_1$.
It will be easiest to work with angles in the source unfolding. So
\begin{align*}
\a & = \angle y_1 u_{13} a_{13} \\
\b & = \angle y_1 u_{12} a_{12} 
\end{align*}
\begin{figure}[htbp]
\centering
\includegraphics[width=1.0\textwidth]{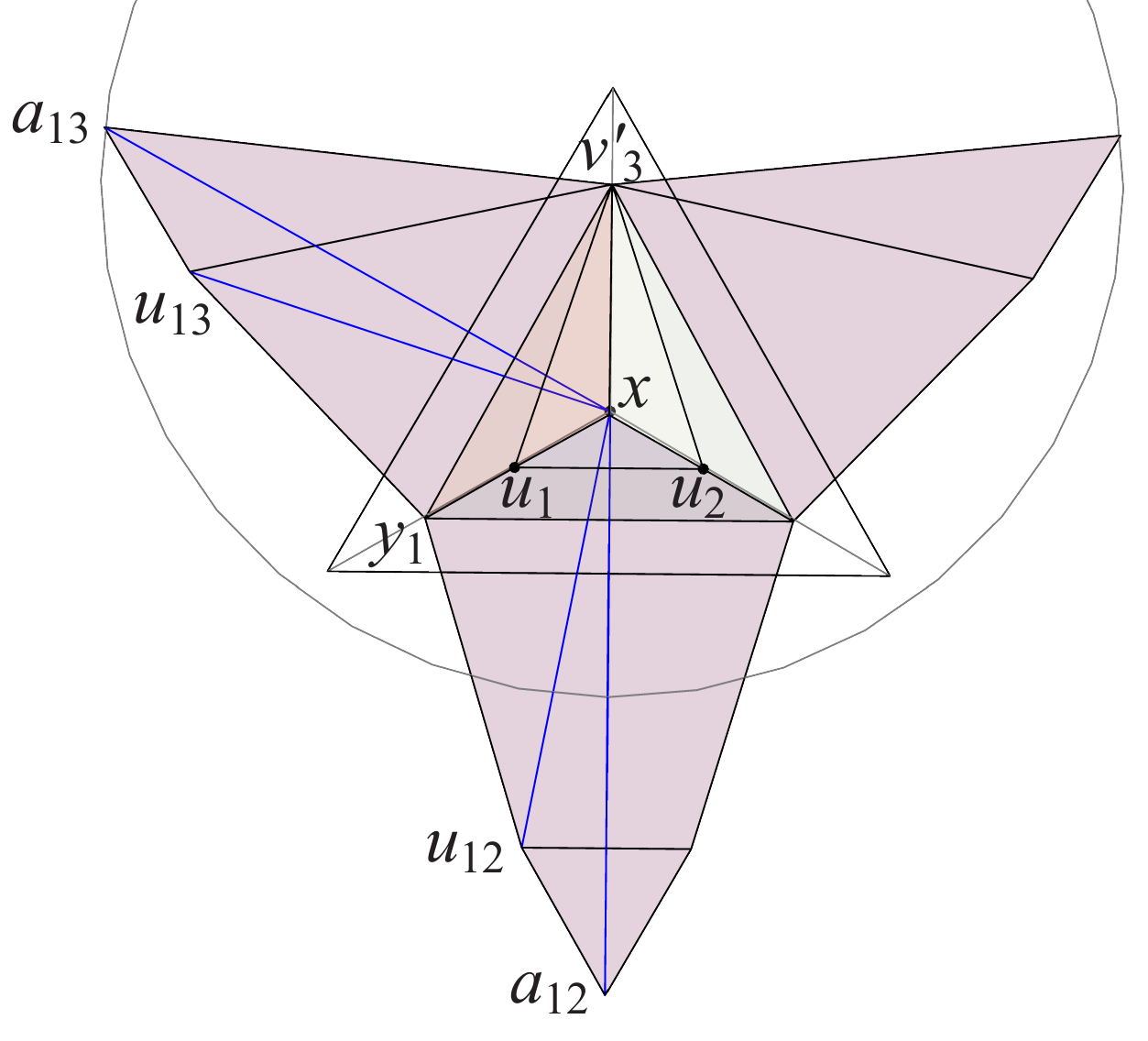}
\caption{Overhead view. $\a=\b$.}
\figlab{Proof_u12_alphabeta}
\end{figure}
Now re-interpret these angles from $x$:
see Fig.~\figref{Proof_u12_alphabeta}:
\begin{align*}
\a & = \angle x u_{13} a_{13} + \angle x u_{13} y_1 \\
\b & = \angle x u_{12} a_{12} + \angle x u_{12} y_1
\end{align*}
By construction (Item~\ref{u12})
\begin{align*}
\angle x u_{13} a_{13}  = \angle x u_{12} a_{12}
\end{align*}
and the triangles $x u_{13} y_1$ and $x u_{12} y_1$
are congruent as all their respective edges are congruent.
Therefore $\a=\b$.
\end{enumerate}

Fig.~\figref{Proof_u12_12} shows the completed construction
for $u_1 = 0.7 a + 0.3 v_1$.
Note that $y_1 = L_{13} \cap L_{12}$  does not necessarily lie on $x v_1$.
And note that $L_{12}$ is parallel to $v_1 v_2$, but
$L_{13}$ is not parallel to $v_1 v_3$.
These details may be more evident in Fig.~\figref{Proof_u12_13_04},
when $u_1 = 0.4 a + 0.6 v_1$.
In a sense, this lack of ``symmetry'' reflects the fact that in
Case~(c), the polyhedron edges left and right of $u_1$ have different destinations:
to one side terminating on the base, to the other side terminating at $u_2$.

\begin{figure}[htbp]
\centering
\includegraphics[width=1.0\textwidth]{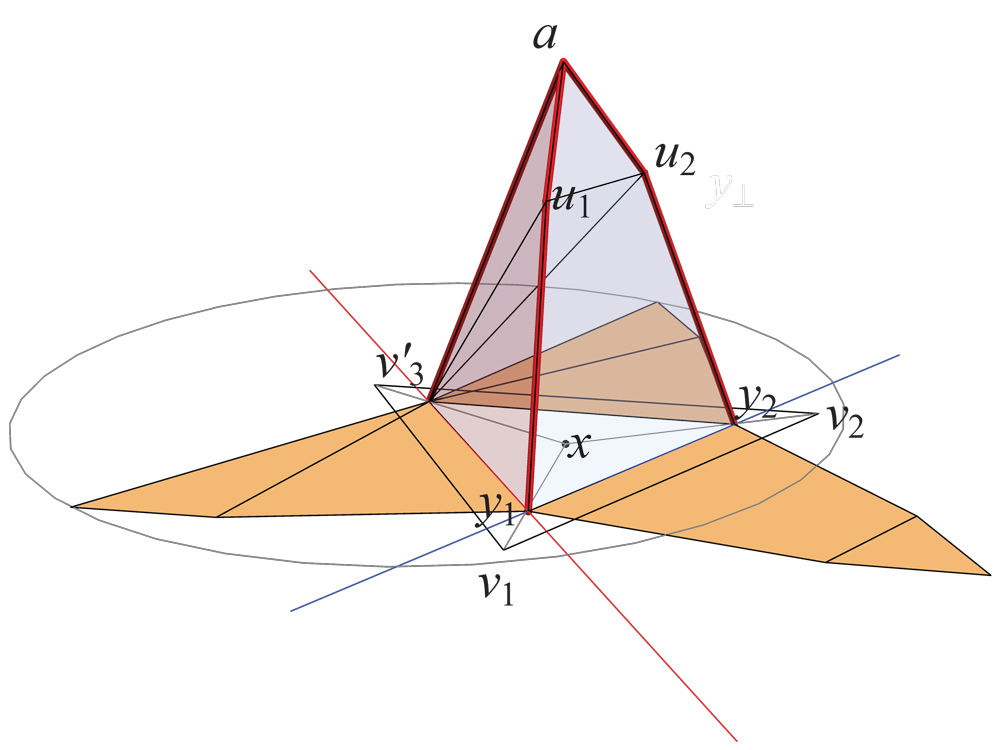}
\caption{Final construction.
$u_1,u_2$ at  $70$\% of $a v_1$}
\figlab{Proof_u12_12}
\end{figure}

\begin{figure}[htbp]
\centering
\includegraphics[width=1.0\textwidth]{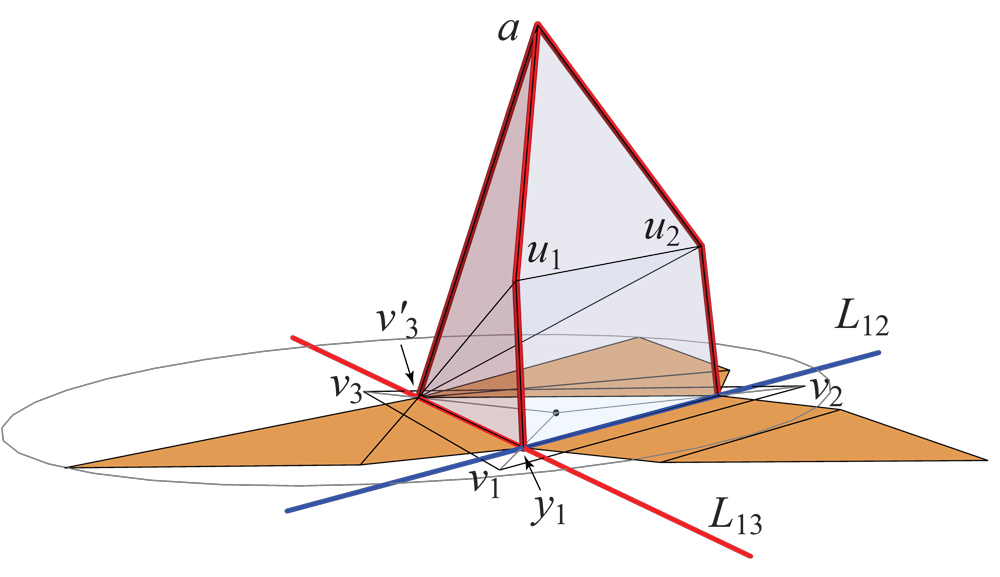}
\caption{$u_1,u_2$ lower: $40$\% of $a v_1$.}
\figlab{Proof_u12_13_04}
\end{figure}

\bigskip


\begin{rmk}
\rmklab{About Case (c)}
Assume that, in Fig.~\figref{FourCases}(c), instead of only one blue edge $u_1u_2$, 
there is a chain $C$ of several blue (horizontal) edges between $u_1$ and $u_2$, 
separated by red branches of $\C(x)$.
Then the above construction still works:
we first apply it for the two blue edges incident to $u_1$: a slanted one and a horizontal one;
afterward we iterate only its second part, starting with Item~\ref{a-12}.
%
%
%
\end{rmk}


\subsection{Case~(d)}
\seclab{Case_d}
Recall that Case~(d) occurs when a chain of degree-$2$ nodes, connected by
polyhedron edges (blue in Fig.~\figref{FourCases}(d)),
closes to itself, as in the polyhedron in Fig.~\figref{StackedPyramids}.

For simplicity of the exposition, we start with a tetrahedron $T=a v_1 v_2 v_3$
and with the point $x \in v_1v_2v_3$ realizing a skeletal cut locus $\C(x)$.
(The case of arbitrary pyramids is analogous.)

Our goal is to modify $T$
by adding several nodes of degree-$2$ on ``consecutive'' branches of
$\C(x)$, consecutive in the sense that they end at consecutive leaves $v_i$.

\medskip

\noindent\textbf{Step~1}. First we show how to modify $T$ in order to obtain a degree-$2$ node $u_1$, and later (Step~2) 
we explain how these modifications are compatible with
modifications necessary to obtain other degree-$2$ nodes, on other branches
of the given tree.
(Step~(1) shares some similarity to Case~(b), Section~\secref{Case_b}.)

\begin{enumerate}
\item 
We focus on the edge $a v_1 \subset \C(x)$.
Denote by $\d$ the distance on $T$ from $x$ to $a$,
which is realized by 
three geosegs, each crossing one lateral face.
\item 
Consider a point $y_1$ on the line segment $xv_1$, close to $v_1$.
Construct through $y_1$ the lines $L_{12} \parallel v_1v_2$ and $L_{13} \parallel v_1v_3$.
See Fig.~\figref{Cased_1}.
\begin{figure}[htbp]
\centering
\includegraphics[width=0.75\textwidth]{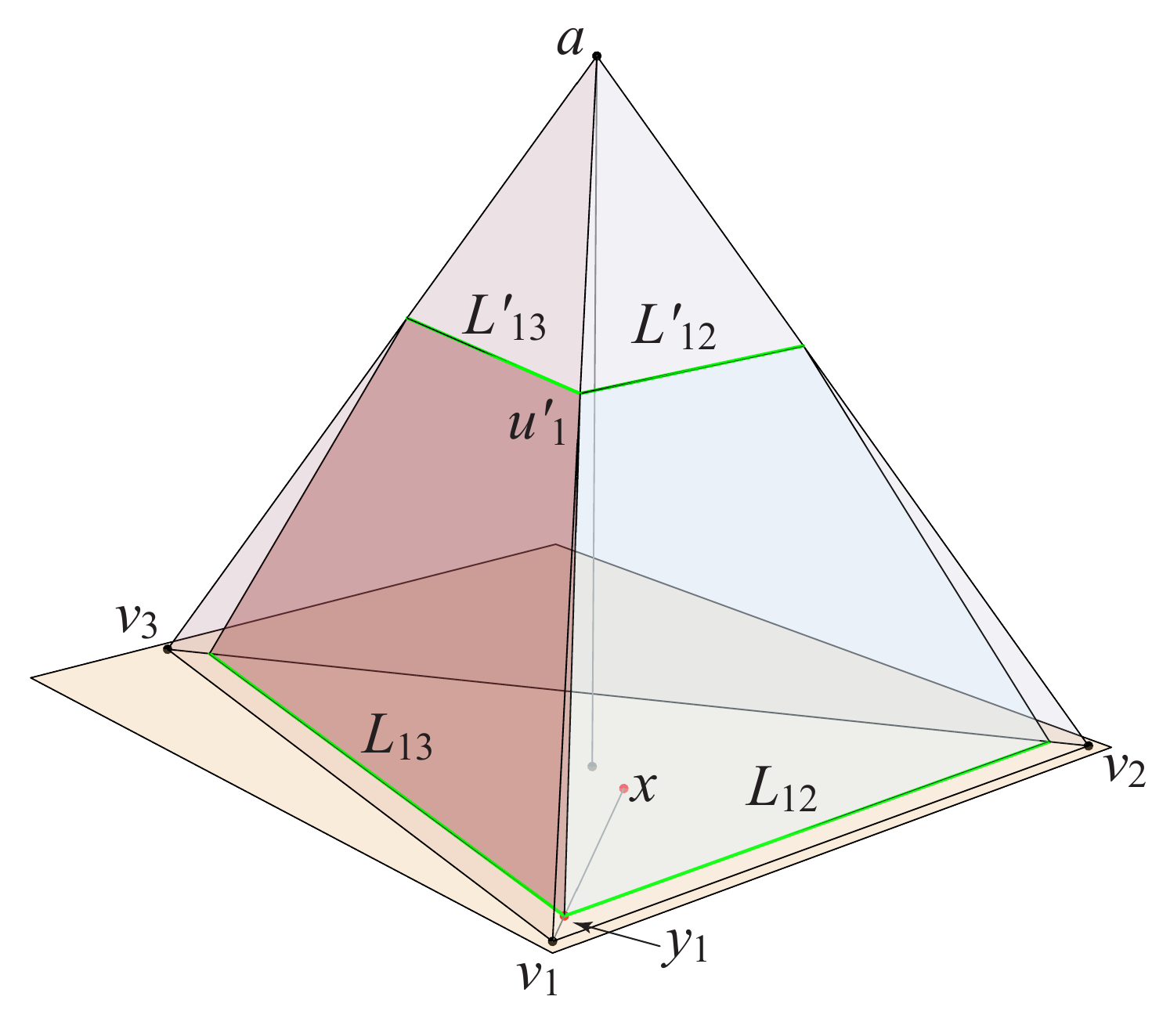}
\caption{$P'$: $T$ truncated by planes through 
(green) lines:
$L_{12} \cup L'_{12}$, and 
$L_{13} \cup L'_{13}$.}
\figlab{Cased_1}
\end{figure}
%
\item 
Consider a point $u'_1$ on the edge $a v_1$.
Construct through $u'_1$ the lines $L'_{12} \parallel v_1v_2$ and $L'_{13} \parallel v_1v_3$.
\item 
Truncate $T$ with the two planes determined by $L_{12} \cup L'_{12}$, and 
by $L_{13} \cup L'_{13}$.
Denote by $P'$ the resulting polyhedron.

Clearly, there are two geoarcs on $P'$ from $x$ to $a$,
a ``right'' one and a ``left'' one.
Because $L_{12}, L_{13}$ are inside $T$, both those geoarcs are smaller than $\d$.
\item
Now the plan is (informally) to move  $L'_{12}$ rightward until the modified right faces increase the right geoarc's distance from $x$ to $a$ to match $\d$, and then to move 
$L'_{13}$  leftward to achieve the same distance $\d$ for the left geoarc.
\item 
Consider a variable line $\D_z$ displacing continuously 
by a horizontal offset $z$ from its initial position $L'_{12}$ toward the exterior of $P'$, 
maintaining at all times $\D_z \parallel L_{12}$.

Denote by $\d_z$ the length of the shortest path from $x$ to $a$ which crosses $L_{12}$ and $\D_z$.
Clearly, $\d_z$ increases continuously from some value $< \d$ to arbitrarily large values.
Therefore, there exists a position of $D_z$ for which $\d_z=\d$.
Denote by $\D_{12}$ this position.
See Fig.~\figref{Cased_2}.
\begin{figure}[htbp]
\centering
\includegraphics[width=1.0\textwidth]{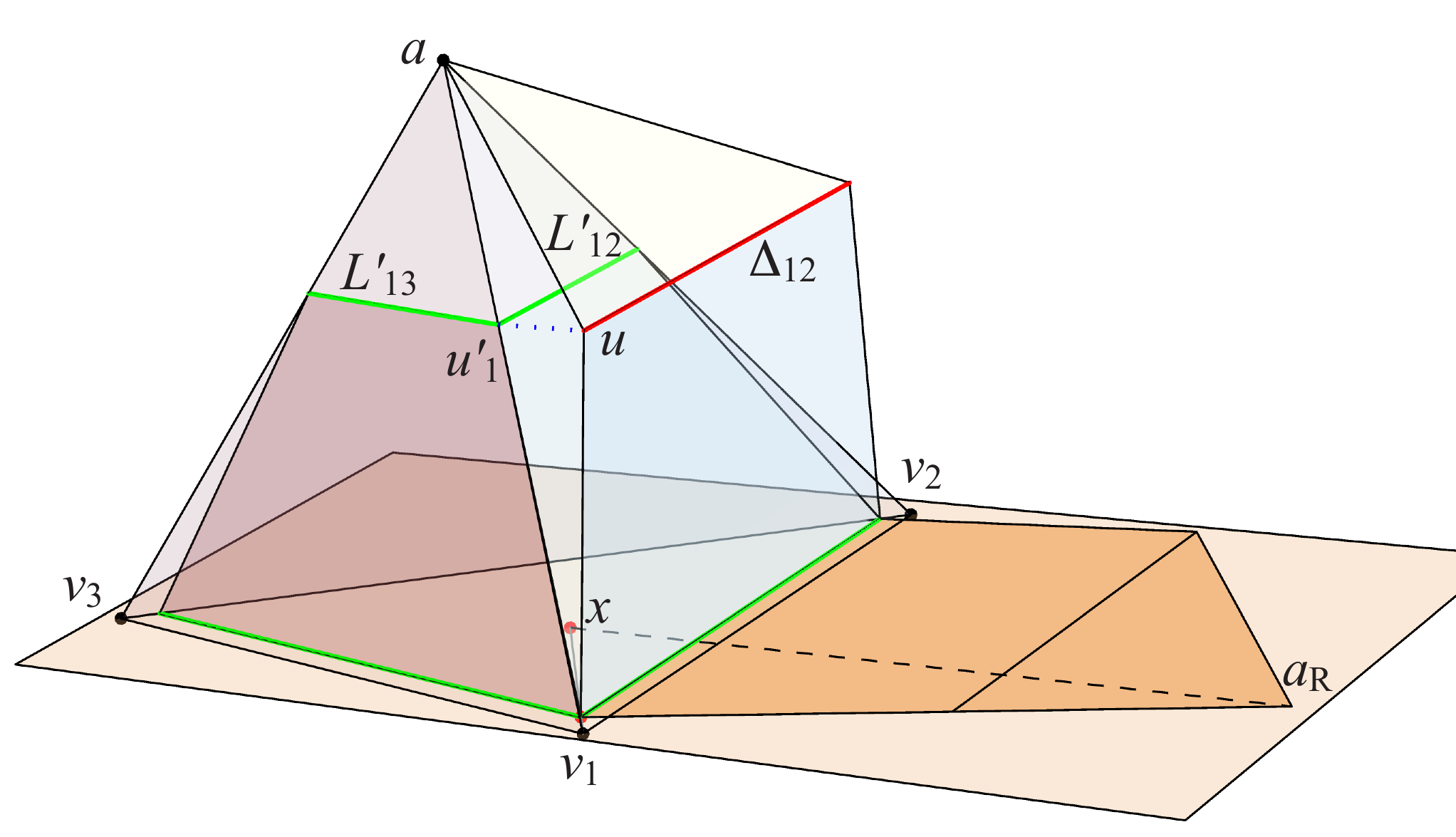}
\caption{$L'_{12}$ (green) $\to \D_{12}$ (red), when $\d_z=\d$.
The geoseg from $x$ to $a$ across $\D_{12}$ unfolds to straight segment $x a_R$
in the base plane.}
\figlab{Cased_2}
\end{figure}
%
\item 
\label{var u}
Move continuously a variable point $u$ along $\D_{12}$, and consider the line
$\D_u \parallel L_{13}$ through $u$.

The starting position of $u$ is when $\D_u$ is the supporting line of $P'$ closer to $L_{13}$.
The point $u$ displaces toward the exterior of $P'$,
maintaining at all times $\D_z \parallel L_{13}$.

Denote by $\d_u$ the length of the shortest path from $x$ to $a$ which crosses $L_{13}$ and $\D_u$.
Clearly, $\d_u$ increases continuously from some value $<\d$ to arbitrarily large values.
Therefore, there exists a position of $D_u$ for which $\d_u=\d$.
Denote by $\D_{13}$ this position.
Note that the right faces of $P'$ are unaltered by the movement of $\D_u$,
so $\d_z=\d$ still holds.
See Fig.~\figref{Cased_3}.
\begin{figure}[htbp]
\centering
\includegraphics[width=1.0\textwidth]{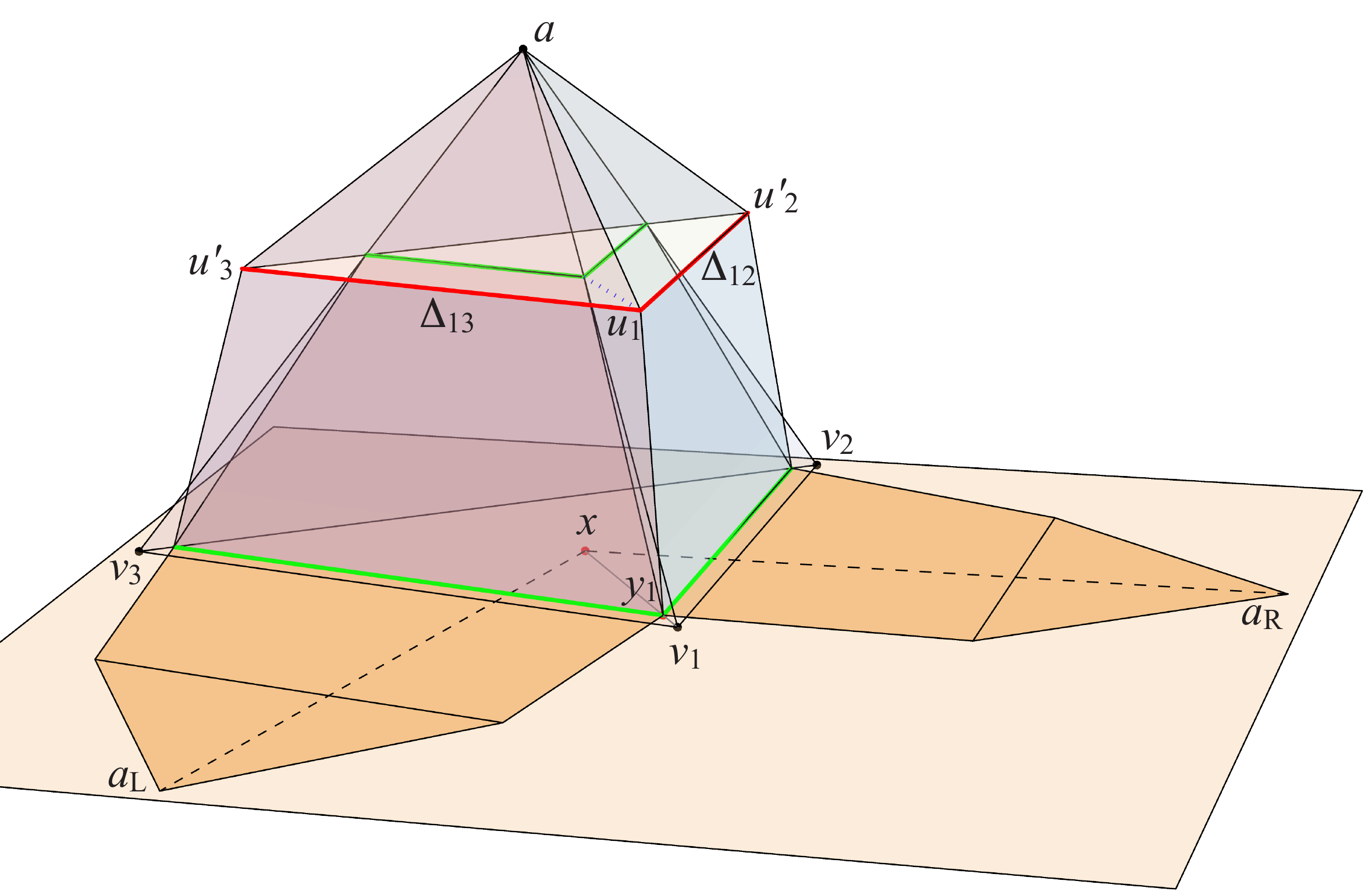}
\caption{$u_1 = \D_{12} \cap \D_{13}$.}
\figlab{Cased_3}
\end{figure}

%
\item
Let $u_1$ denote the intersection point of $\D_{12}$ and $\D_{13}$.
Let $B$ be the plane containing the back face $a v_2 v_3$ of $T$.

Put $\{u'_2\} = \D_{12} \cap B$, $\{u'_3\} = \D_{13} \cap B$,
$\{y'_2\} = L_{12} \cap B$, $\{y'_3\} = \D_{13} \cap B$.
See Fig.~\figref{Cased_4}.
\begin{figure}[htbp]
\centering
\includegraphics[width=1.0\textwidth]{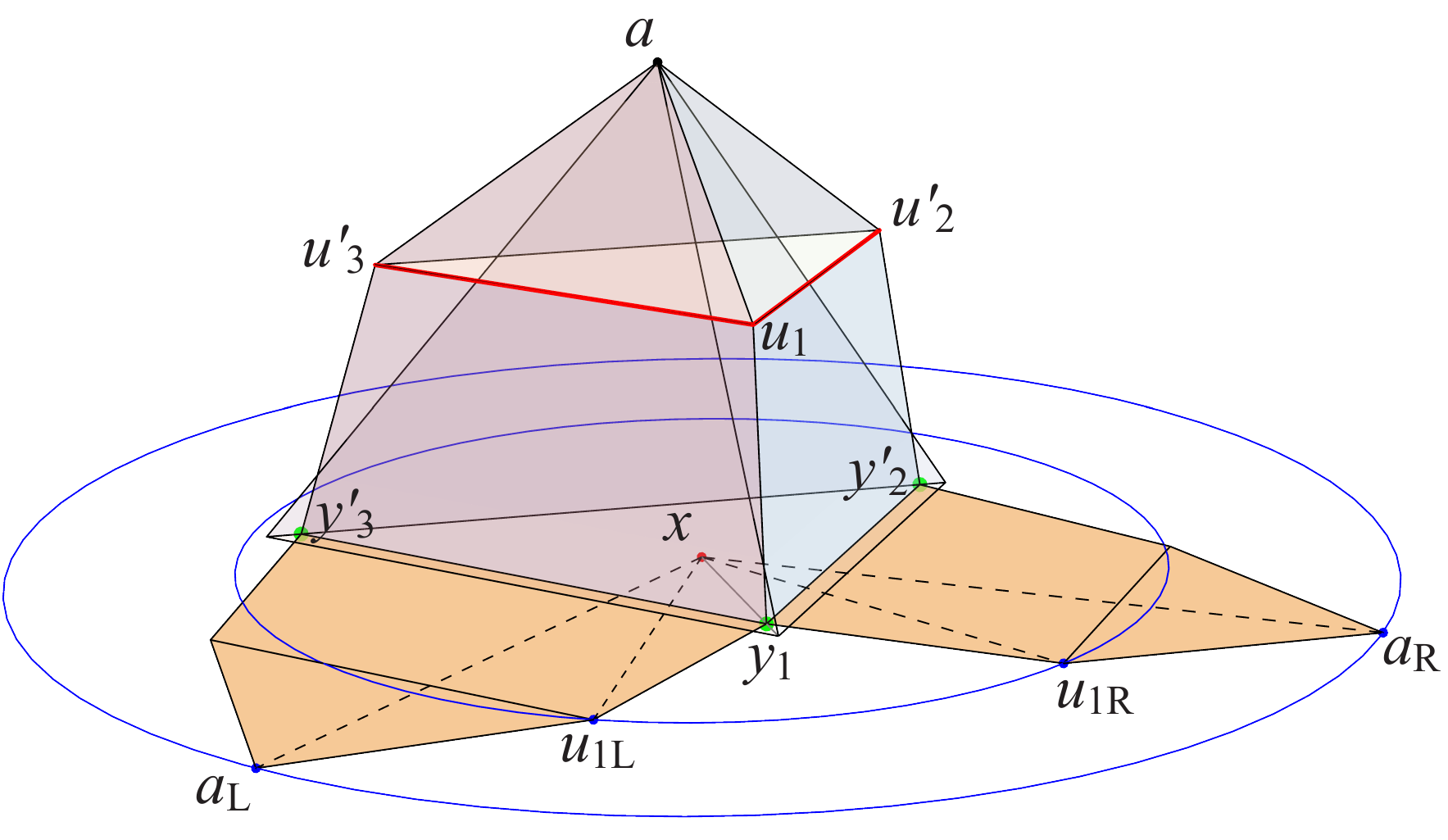}
\caption{$P=\conv(a, u_1, u'_2, u'_3, y_1, y'_2, y'_3)$.
Unfolded images of $x a$ and $x u_1$ shown dashed.}
\figlab{Cased_4}
\end{figure}
%
Denote by $P$ the convex polyhedron with vertices $a, u_1, u'_2, u'_3, y_1, y'_2, y'_3$. 
Notice that $P$ converges to $T$ if $y_1$ converges to $v_1$.

Let ${\rm dist}^Q(x,p)$ be the distance from $x$ to point $p$ on polyhedron $Q$.
To this stage of the argument, we have established that
the point $u_1$ is exterior to $T$ such that, on $P$:
\begin{itemize}
\item ${\rm dist}^P(x,a)$ is equal to the distance $\d$ on $T$ from $x$ to $a$---${\rm dist}^T(x,a)$---and
\item ${\rm dist}^P(x,a)$ is obtained by three geosegs: one
crossing $\D_{12}$, one crossing $\D_{13}$, and one up the back face $B$,
which derives from $T$ and has remained unaltered by all previous changes.
\end{itemize}
We thus have shown that $a \in \C(x)$.
\item Unfolding $T$ and $P$ in the base plane gives the same images $a_L$, $a_R$ of $a$,
by the choice of $u_1$. See Fig.~\figref{Cased_4}.
We next turn to showing that $u_1 y_1$ is in $\C(x)$,
by considering the source-unfolding images $u_{1R},u_{1L}$.
\item The planar triangles $xv_1a_R$ and $xv_1a_L$ are congruent (all sides equal), hence $\angle y_1 x a_R = \angle y_1 x a_L$.
\item 
\label{equal angles 1}
Therefore the triangles $y_1 x a_R$ and $y_1 x a_L$ are congruent, hence 
$|y_1 a_R|=|y_1 a_L|$ and $\angle x y_1 a_R = \angle x y_1 a_L$.
\item 
\label{equal angles 2}
Therefore, the triangles $y_1 u_R a_R$ and $y_1 u_L a_L$ are congruent (all sides equal). So $\angle a_R y_1 u_R = \angle a_L y_1 u_L$.
\item From \ref{equal angles 1} and \ref{equal angles 2} we have 
$\angle x y_1 u_R =\angle x y_1 u_L$.
Therefore, $y_1u_1$ is contained in $\C(x)$ on $P$.
\item Therefore, on $P$, $|x u_R| = |x u_L|$. Because $a \in \C(x)$ on $P$, 
$a u_1 \subset \C(x)$ on $P$.
\end{enumerate}

This completes Step~1:
both $a u_1$ and $u_1 y_1$ are in $\C(x)$, and $u_1$ is of degree-$2$ in $\C(x)$.

\bigskip

\noindent\textbf{Step~2}.
We now discuss the compatibility of the above changes around $v_1 /y_1$ with other changes.

We take $\{y_2\}= L_{12} \cap x v_2$, and we iterate the above construction from Item~\ref{var u} onward. 
So we identify the point $u_2 \in \D_{12}$ and the line $\D_{23}$.

Finally, we take $\{y_3\}= L_{23} \cap x v_3$, and we iterate the above construction from Item~\ref{var u} onward. This identifies the point $u_3 \in \D_{23}$.

It remains to prove that, by the above mentioned iterations, we close the chain
of base edges joining $y_i$ and horizontal edges joining $u_i$
(as in Fig.~\figref{StackedPyramids}).

\medskip

First we show that the chain of base edges joining $y_i$ closes.
This follows directly from the next two pairs of similar triangles, with the same similarity ratio:
$x v_1 v_2$ and $x y_1 y_2$, $x v_2 v_3$ and $x y_2 y_3$. 
Consequently, the triangles $x v_1 v_3$ and $x y_1 y_3$ 
are also similar with the above similarity ratio.

\medskip 

Now we see that the chain of horizontal edges joining $u_i$ closes.
This follows from the fact that the lengths of the 
geosegs on the resulting polyhedron, 
from $x$ to $a$ and crossing the respective edges, are all equal to $\d$.

The proof for Case~(d) is thus complete.

\bigskip

Call the polyhedra obtained by successively applying a finite sequence of our constructions \emph{tapered polyhedra}.


\section{Induction Proof}
\seclab{Induction Proof}
With the four degree-$2$ cases settled, we can prove our main theorem.

\setcounter{thm}{0}
\begin{thm}
\thmlab{deg-2}
Given any combinatorial tree $\T$ 
there is a convex polyhedron $P$ and a point $x \in P$
such that the cut locus $\C(x)$ is entirely contained in $\Sk(P)$,
and the combinatorics of $\C(x)$ match $\T$.
\end{thm}

\begin{proof}
If all nodes of $\T$ have degree-$2$ then it can be realized on a doubly covered polygon, see e.g. \cite[Lem.2.2]{Reshaping}. 

So we may assume that $\T$ has at least one node of degree $\ge 3$, say $a$. 
Fix the root of $\T$ at $a$; this will become the apex of the realizing polyhedron.

The proof is constructive, by induction over the 
discrete distance (i.e., the number of edges) to $a$ in $\T$.
Precisely, denote by $\T_k$ the subtree of $\T$ consisting of all nodes at distance at most $k$ from $a$, together with all edges joining them in $\T$.
We show by induction that all $\T_k$ can be realized as cut loci.

For $k=1$, $\T_1$ has one internal node $a$ and as many leaves as the degree of $a$ in $\T$. 
Assume, for the simplicity of the exposition, that $\deg a =3$. (The case $k>3$ can be treated analogously.
See Fig.~\figref{AbstractInduction_2} for an abstract example).

We realize $\T_1$ on a regular tetrahedron $P_1$, 
with $a$ the top apex and $x$ the center of the base $v_1 v_2 v_3$.

Assume now that we have realized $\T_k$ as a cut locus 
on a tapered polyhedron $P_k$
with $a$ the top apex and $x$ the center of the base $v_1 v_2 v_m$, 
where $m$ is number of leaves of $T_k$.
By the induction hypothesis,
$P_k$ is obtained from $P_1$ by successive modifications detailed by our case constructions.

We must manage the interactions between the newly added edges to create degree-$2$ and degree $3$ nodes. Therefore, we start with the creation of degree-$2$ nodes.
We illustrate some cases in Fig.~\figref{AbstractInduction_1}.
\begin{enumerate}[(1)]
\item If all branches have nodes of level $k$, and at least one of them is of degree-$2$, then we start by applying Case~(d).
(This occurs with nodes $2,3,4$ in Fig.~\figref{AbstractInduction_1}(a), 
when Case~(d) is applied).
This will create the necessary degree-$2$ nodes; moreover, all nodes to become of degree $\geq 3$ are now degree-$2$ nodes.
After that we apply other necessary changes, to transform 
some of the nodes of degree $2$ into nodes of higher degree, by appropriate truncations described in Section~\secref{ConstructionDetails}. 
\item Otherwise, if at least one branch has no node of level $k$ and there are nodes of degree-$2$, we apply Cases~(a)-(c) whichever ones are appropriate
(This occurs with nodes $5,6$ in Fig.~\figref{AbstractInduction_1}(a), when Case~(a) is applied).
And finally,
\item If all branches have nodes of level $k$ and degree $\geq 3$ we apply the respective construction.
\end{enumerate}

\medskip
This completes the induction proof and establishes Theorem~\thmref{deg-2}.
\end{proof}

\begin{figure}[htbp]
\centering
\includegraphics[width=1.0\textwidth]{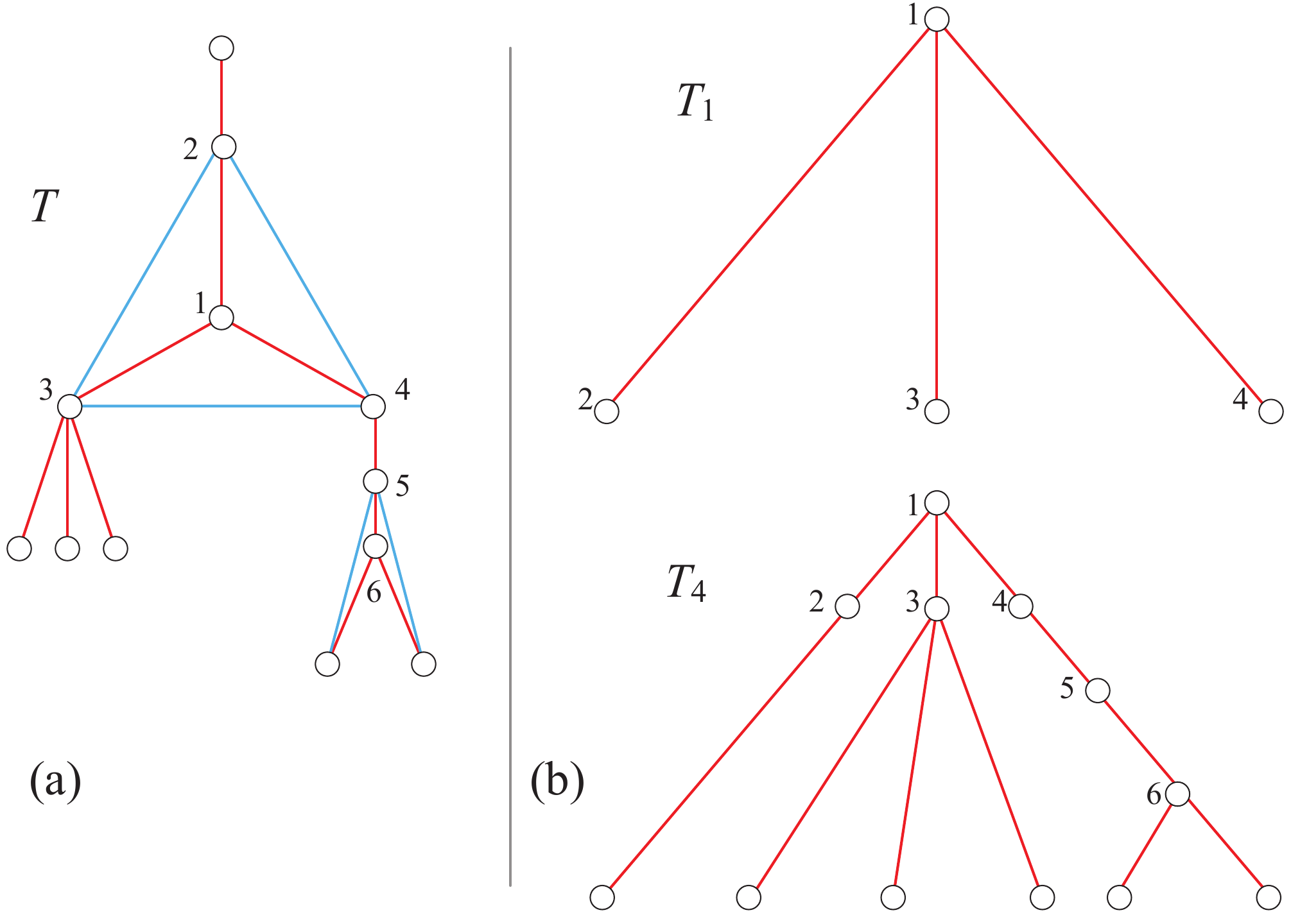}
\caption{(a)~Tree $\T$. Blue edges are non-$\C(x)$ polyhedron edges.
Case~(d) is applied to nodes $1,2,3,4$.
Case~(a) is applied to nodes $5,6$.
(b)~The first and last induction steps. 
}
\figlab{AbstractInduction_1}
\end{figure}

\bigskip
\noindent
We also illustrate the induction steps in a second example,
shown in Fig.~\figref{AbstractInduction_2}.
The steps in this example are as follows:
\begin{enumerate}[(a)]
\item The introduction of nodes $2$ and $3$ creates an instance of Case~(c),
because node $2$ is degree-$2$.
\item Node $4$ creates an instance of 
the construction in Section~\secref{ConstructionDetails}. 
\item Node $5$ creates an instance of Case~(b).
\item Nodes $6$ and $7$ constitute an instance of Case~(c).
\end{enumerate}

\begin{figure}[htbp]
\centering
\includegraphics[width=0.65\textwidth]{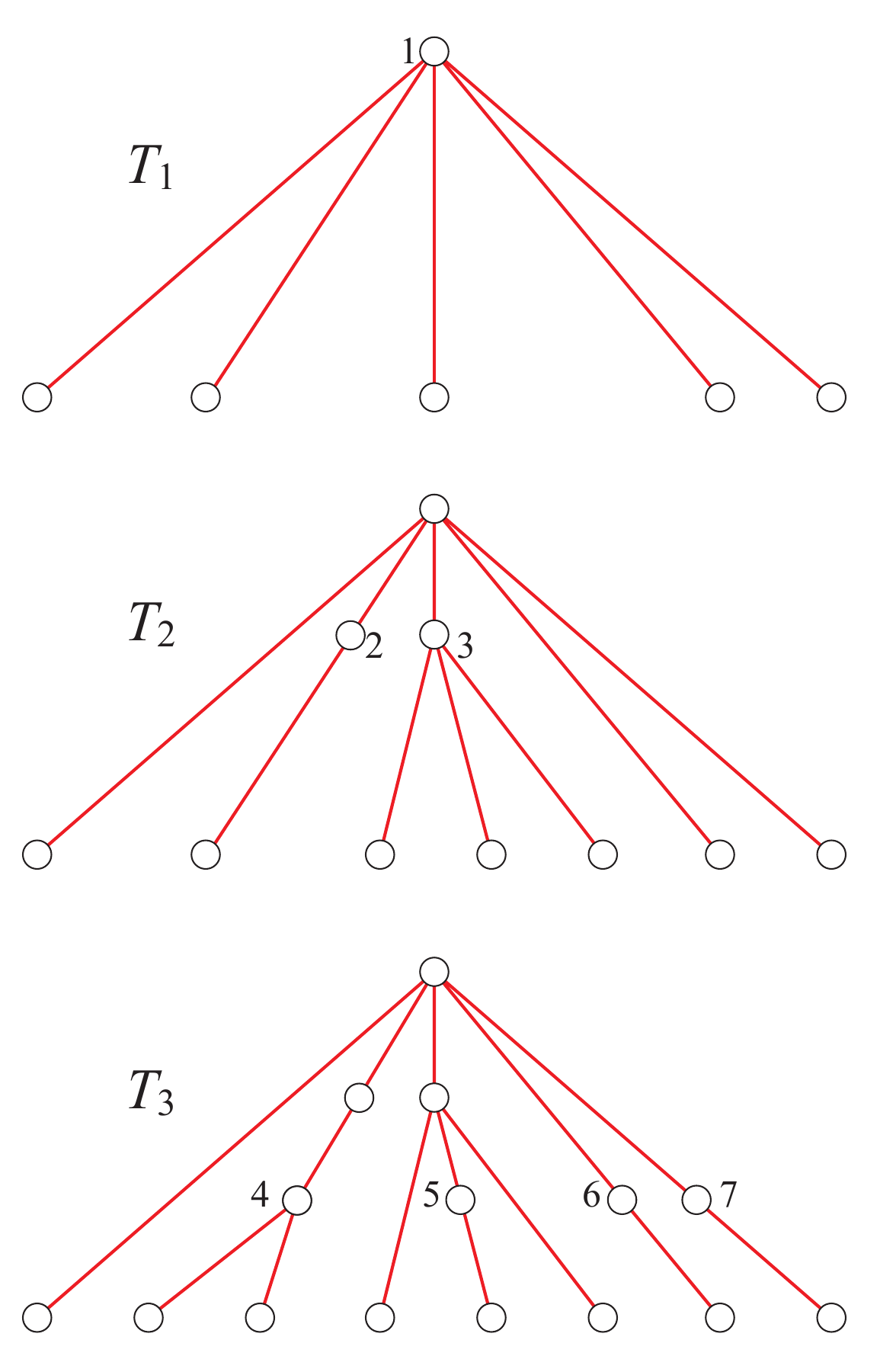}
\caption{Example of induction construction.
Two nodes $2,3$ are added in the transition from $\T_1$ to $\T_2$,
degree-$2$ and degree-$4$ respectively.
Four nodes are added in the transition from $\T_2$ to $\T_3$:
$4$ of degree-$3$ and $5,6,7$ of degree-$2$.
}
\figlab{AbstractInduction_2}
\end{figure}

\bigskip

As our goal has been proving Theorem~\thmref{EveryTree},
we have focussed on the existence of a realizing polyhedron $P$ for any given tree $\T$.
We have not addressed the algorithmic question of actually constructing $P$ from $\T$.
However, because all of our constructions are explicit
(and have been individually implemented), the proof in some sense
already constitutes an algorithm.
Without analyzing it carefully, we expect that the complexity of the algorithm is
proportional to the number of nodes in $\T$, with some data-structure overhead.
We leave establishing this formally to future work.


\section{The Class of Tapered Polyhedra}
\seclab{Tapered}
As we noted in the Introduction, a polyhedron with a skeletal cut locus
leads directly to an edge-unfolding to a net (a nonoverlapping polygon in the plane).
This is because the source unfolding from $x$, which is known to be a net, is achieved by
cutting the cut locus $\C(x)$, which maps to the outer boundary of the unfolding.
Despite considerable effort by researchers to resolve D\"urer's problem~\cite{o-dp-13}---whether or
not every convex polyhedron has an edge-unfolding to a net---there are only a few
infinite classes of polyhedra known to edge-unfold to a net. One class is the
\emph{domes}, polyhedra with a distinguished base face $B$ such that every other face
shares an edge with $B$~\cite[Sec.~25.5.2]{do-gfalop-07}.
A slight extension to \emph{g-domes} (generalized domes) allows a face to share
just a vertex with $B$~\cite[Sec.~3.1]{Reshaping}.

The main theorem of this paper leads to what we call
\emph{tapered polyhedra}, a class that properly includes
some g-domes (and therefore domes).
It is therefore of interest to list a few geometric characteristics of tapered polyhedra:

\noindent
\begin{itemize}
\item Every non-base vertex projects orthogonally to be strictly inside the base.\footnote{
Basically, this was our choice to simplify the reasoning; 
see for example Section~\secref{ConstructionDetails} and Fig.~\figref{Tfour_3d}. 
Avoiding this choice would lead to technical difficulties we have not addressed.}
\item Case~(a) in Fig.~\figref{FourCases} already goes beyond domes, to g-domes:
Two triangle faces share just a vertex with the base.
\item 
A tempered polyhedron is partitioned into zero or more
``level rings'' of nodes deriving from Case~(d), connected
in a cycle of vertices, as (for example) in Fig.~\figref{StackedPyramids}.
These rings take us beyond domes and g-domes, and so the tapered polyhedra
constitute a new class of polyhedra with edge-unfoldings to nets.
\end{itemize}

\noindent
We do not yet know a complete geometric characterization of tapered polyhedra.

We mention here that the order of assembling the cases to prove (by induction) Theorem~
\thmref{EveryTree} may differ. What we proposed in Section~\secref{Induction Proof} is just one viable way, among perhaps several others.
And the order of assembling the cases may result in different subclasses of polyhedra, 
within the class of tapered polyhedra.
%

\bigskip
We conclude by repeating this central open problem from Section~\secref{Introduction}.
\begin{op}
Characterize all the polyhedra $P$ that support a skeletal cut locus,
i.e., characterize the cut locus amenable polyhedra.
\end{op}

\clearpage
\bibliographystyle{alpha}
\newcommand{\etalchar}[1]{$^{#1}$}

\end{document}